\documentclass[journal,10pt]{IEEEtran}

\usepackage{amssymb}
\usepackage{amsmath}
\usepackage{amsthm}
\usepackage{color}
\usepackage{subfigure}

\usepackage[ruled,linesnumbered]{algorithm2e} 

\usepackage[pdftex]{graphicx}
\graphicspath{{../pdf/}{../jpeg/}}
\usepackage{epstopdf}

\definecolor{red}{rgb}{1.0,0.0,0.0}
\newtheorem{lemma}{Lemma}

\begin{document}

\title{Spectral Efficient and Energy Aware Clustering in Cellular Networks}

\author{Georgios~Kollias,~Ferran~Adelantado,~and~Christos~Verikoukis
}


\maketitle

\begin{abstract}
The current and envisaged increase of cellular traffic poses new challenges to Mobile Network Operators (MNO), who must densify their Radio Access Networks (RAN) while maintaining low Capital Expenditure and Operational Expenditure to ensure long-term sustainability. In this context, this paper analyses optimal clustering solutions based on Device-to-Device (D2D) communications to mitigate partially or completely the need for MNOs to carry out extremely dense RAN deployments. Specifically, a low complexity algorithm that enables the creation of spectral efficient clusters among users from different cells, denoted as enhanced Clustering Optimization for Resources' Efficiency (eCORE) is presented. 
Due to the imbalance between uplink and downlink traffic, a complementary algorithm, known as Clustering algorithm for Load Balancing (CaLB), is also proposed to create non-spectral efficient clusters when they result in a capacity increase. Finally, in order to alleviate the energy overconsumption suffered by cluster heads, the Clustering Energy Efficient algorithm (CEEa) is also designed to manage the trade-off between the capacity enhancement and the early battery drain of some users. Results show that the proposed algorithms increase the network capacity and outperform existing solutions, while, at the same time, CEEa is able to handle the cluster heads energy overconsumption.       
\end{abstract}

\begin{IEEEkeywords}
Cellular Networks, Clustering, Device-to-Device, Traffic Imbalance.
\end{IEEEkeywords}

\IEEEpeerreviewmaketitle

\section{Introduction}
\label{S:intro}

The envisaged increase of the cellular traffic, which according to \cite{CISCO} is expected to reach 30.6 exabytes per month by 2020 at a compound annual growth rate (CAGR) of 53\%, imposes new capacity challenges to the fifth generation (5G) cellular networks. Specifically, this ever-increasing trend in data traffic demand will force 5G networks 
to meet a 1000$\times$ capacity increase, mainly based upon three pillars: the improvement of the spectral efficiency, the allocation of new spectrum bands, and the densification of the Radio Access Network (RAN) \cite{5Gsurvey}. 
Focusing on the densification of the RAN, the research community has proposed the dense deployment of Small Cells (SC) as an enabler 
for the capacity increase required to meet the expected traffic demand. However, such densification of the RAN has posed significant technological challenges, such as interference management \cite{Top_manag1}\cite{Top_manag2}, and economic considerations. The high Capital Expenditure (CAPEX) and Operational Expenditure (OPEX) incurred by the Mobile Network Operators (MNO) when densifying the network hamper the actual deployment of (ultra-)dense RANs \cite{SCs_Cost}.

As mobile devices are the main contributors to the traffic growth, high capacity demand is intrinsically linked to the boost in the number of mobile devices connected to the network. For instance, and based on \cite{CISCO}, the number of mobile devices and connections will globally reach 11.6 billion by 2020. 
Therefore, the need for denser RAN deployments run in parallel with the actual and envisaged growth of the density of mobile devices.
In this context, where the densification of the network is jeopardized by the high deployment costs, we propose the exploitation of the cooperation among mobile devices (through Device-to-Device communications, D2D~\cite{22803}) as a cost-efficient solution to expand the RAN when and where needed. The inclusion of mobile devices as an expansion of the RAN can provide high spatial diversity and improve the spectral efficiency of the whole network.
Although cooperation among Base Stations (BS) has already been proposed as a mean to increase the spectral efficiency (e.g. \cite{Dynamic_clust}), cooperation among devices proposed in the sequel opens up new opportunities and challenges to get the network dynamically adapted to traffic needs.

The rest of the paper is organized as follows: The State of the Art and the contributions of the proposal are detailed in Section \ref{S:sota}. In Section \ref{S:proposal} the system is modelled as an optimization problem, and two clustering algorithms, namely enhanced Clustering Optimization for Resources Efficiency (eCORE) and Clustering algorithm for Load Balancing (CaLB), are presented. Section \ref{S:energy} analyses the energy consumption challenges and proposes a Clustering Energy Efficient algorithm (CEEa) to prevent cluster heads from early battery drain. Finally,  numerical results are presented in Section \ref{S:results}, while concluding remarks are given in Section \ref{S:conclusions}.

\section{State of the Art and Contributions}
\label{S:sota}

The need to improve spectrum utilization, overall throughput and energy consumption in cellular networks has stimulated the research on the D2D field over the last years.  .
In short, D2D communications are expected to become the basis to provide direct connectivity between users (with or without the support of network infrastructure), enable devices to play the role of relay in two-hops communications with the BS, and allow the multicast of common content from the BS to a multicast group via a forwarding cluster head user \cite{SURVEY}.

Regarding the direct connectivity between two users, \textit{Feng et al.} proposed in \cite{DEVICE_TO} a resources' allocation framework to optimize the spectral efficiency of the network when a set of D2D pairs underlying the cellular network operate over the same frequency as the cellular users. In this study, however, D2D pairs are never connected to the cellular network and therefore the D2D pairs have only two options: transmit in D2D mode or remain silent. Similarly, \cite{Clustering4} analyses the joint power control and frequency reuse of D2D pairs in the same scenario presented in \cite{DEVICE_TO}. cellular users.
Also in line with \cite{DEVICE_TO} and \cite{Clustering4},  \textit{J. Huang et al.} proposed in \cite{intercell} a significant step towards more efficient D2D communications by expanding these communications from intra-cell environments to inter-cell environments. The proposal, which is based on game theory, shows clearly the potential of this inter-cell cooperation.
Yet, the scenario is restricted to a very specific use case, characterized by disjoint sets of D2D users and cellular users.

The works in \cite{Clustering1}-\cite{Clustering3} study the performance of D2D communications in multicast groups, where all users download a common content from the BS via a cluster head user. It is shown that a better efficiency in the resources' usage can be achieved in these scenarios, although the gain is always bounded by the lowest quality link between the cluster head and the rest of users of the multicast group. 
In more detail, the authors in \cite{Clustering1} derive expressions in order to select the optimal number of D2D retransmitters in a multicast group, and \cite{Clustering2} proposes a Conventional Multicast Scheme (CMS) to decide whether a user should be served by the BS or by the cluster head.

Similarly, \textit{Meshgi et al.} \cite{multicast_new} maximize the throughput in a single cell scenario with multicast D2D groups underlying a cellular network by proposing a heuristic resource allocation algorithm that achieves near optimal performance. In \cite{Clustering_small} the authors address the multicast clustering by setting up a Primary Cluster Head (PCH) and a Secondary Cluster Head (SCH). In this proposal, the PCH and the SCH are selected based on their residual energy and the received Signal to Interference Noise Ratio (SINR).
Similarly, in \cite{clustering_concept} the authors analyse a set of different strategies for the establishment of multicast clusters. The work shows that D2D-based multicast clustering can increase the system capacity, although it is very sensitive to key parameters, such as clusters' dimension.

Finally, key features required to support network controlled D2D-based multicasting are analysed in \cite{Clustering3}.

Although the works described so far address the problem of clustering in cellular networks with underlying D2D communications, they are limited by the multicast assumptions: i) only downlink traffic is considered and ii) the same content is delivered  to all users in the cluster/multicast group.

Cooperative D2D communications move a step forward in \cite{bench}, where authors formulate the clustering problem as the maximization of the throughput constrained by energy efficiency. The proposed algorithm outperforms the results obtained without clustering but it neglects two important aspects: i) the mobility, that impacts on the quality of the links and on the role played in the cooperation by each user, and ii) the energy consumed by the relay/cluster head when not transmitting, that could be even higher than the energy consumed in transmission state. 

In contrast with the State of the Art, we propose clustering algorithms that are intended to improve the resources' utilization efficiency in a general scenario, where both uplink and downlink traffic are considered in a LTE-A FDD system. The algorithms proposed in the sequel are based on our previous work \cite{globecom}, where the clustering algorithm CORE was proposed. CORE restricted the creation of spectral efficient clusters to users within the same cell thus significantly limiting the achieved gains in dense Heterogeneous Networks (HetNets).  In order to go beyond the aforementioned constraint, we propose a new algorithm, namely enhanced Clustering Optimization for Resources' Efficiency (eCORE), that extends clustering to multi-cell deployments. 
Specifically, eCORE is based on the cooperation among devices by leveraging the D2D communication concept, initially introduced in the framework of LTE-A to support \textit{Proximity-based Services (ProSe)} for public safety \cite{22803}. In the solution proposed hereafter, the mobile devices create spectral efficient clusters with a single cluster head (CH) characterized by good quality links with the serving BS and with the rest of cluster members. In eCORE clusters can be created among users from different cells as long as they result in a decrease of the required resources. The cluster head is responsible for receiving and forwarding packets from/to the BS and the cluster members.  As traffic is more intense in the downlink (DL) and D2D communications are usually carried out over uplink (UL) bands to limit the interference caused to neighbouring users \cite{DEVICE_TO,intercell}, the proposal benefits from the imbalance between uplink and downlink traffic intensity and the high channel gain of D2D communications to increase the capacity of the network. 
Although the dynamic adaptation to the imbalance between uplink and downlink traffic has already been addressed in \cite{TDD2,TDD1} for TDD HetNets, the problem is more challenging in FDD systems, where transferring traffic from the downlink to the uplink is more complex.

Following this rationale, it is shown that the capacity of the network can be further increased by balancing uplink and downlink traffic. In order to benefit from this fact, the Clustering algorithm for Load Balancing (CaLB) is the second proposed algorithm that, run after eCORE, intensifies the clustering process by establishing non-spectral efficient clusters that increase the capacity when downlink and uplink traffics are significantly unbalanced. In particular, CaLB shows that in some cases clustering can be beneficial despite increasing the number of required spectrum resources. Yet, the proposed solutions present challenges in terms of energy consumption of the cluster head that are studied and addressed along the paper by complementing eCORE with the Clustering Energy Efficient Algorithm (CEEa). CEEa is proposed to limit the energy overconsumption experienced by cluster heads and thus minimizing the disincentive in the creation of clusters. Both CaLB and CEEa are algorithms designed to be executed after eCORE to improve its performance (either in terms of capacity or in terms of energy overconsumption), but not to be implemented in a standalone manner.

In a nutshell, the three clustering proposals are designed as a cost-efficient RAN densification solution based on the cooperation of mobile devices in FDD-LTE networks, and the contributions of this work are the following:

\begin{itemize}

\item A new RAN densification solution based on D2D clustering in the framework of FDD LTE-A is presented to improve the spectral efficiency. The algorithm, which is an extension of CORE \cite{globecom} and is denoted by eCORE, exploits the spatial diversity provided by the high density of users and the imbalance between uplink and downlink traffic.  
Contrary to CORE, eCORE enables the creation of clusters among users from different cells. 

\item A load balancing clustering algorithm, namely CaLB, is proposed to increase the capacity of the network. In contrast with eCORE that creates spectral-efficient clusters, CaLB complements eCORE by establishing non-spectral efficient clusters. The capacity gain results from the uplink and downlink load balancing.

\item We propose a complementary algorithm to eCORE, known as CEEa, that compensates the energy overconsumption suffered by cluster heads in eCORE. CEEa benefits from the variation of the scenario caused by mobility and forces reclustering by limiting the time during which users play the role of cluster head to reduce the energy overconsumption.

\end{itemize}

\section{Clustering Proposal}
\label{S:proposal}

The proposed clustering solutions described in the sequel (eCORE, CaLB and CEEa) are all based on a set of premises: i) each cluster has a single cluster head; ii) each user/device can be direcly served by a BS, play the role of cluster head, or join a cluster to be served by a BS through the corresponding cluster head, but no more than a single role can be played simultaneously; iii) intra-cluster communications are D2D transmissions carried out in the uplink band to limit the incurred interference \cite{DEVICE_TO,intercell}. In FDD, the creation of a cluster is translated into a transfer of resources' utilization from the downlink band to the uplink band, which is usually underutilized. For instance, the downlink traffic of a clustered user is first served with downlink resources (from the BS to the cluster head) and subsequently with uplink resources in the D2D communication from the cluster head to the cluster member. If we assume that the channel gain from the BS to the cluster head is higher than the channel gain from the BS to the rest of clustered users (i.e. the cluster head is the user with the best link to the BS among the cluster members), the required downlink resources are reduced.
Although the three algorithms share a set of premises, they differ in their objectives. Thus, in eCORE clustering is aimed to reduce the number of required resources (Section \ref{SS:near_optimal}). In CaLB, the creation of a cluster must decrease the load of the downlink (Section \ref{CaLB}). Finally, in CEEa the energy overconsumption of cluster heads must be compensated (Section \ref{SS:practical_approach}).

This Section is focused on the algorithms that improve the capacity of the network, i.e. eCORE and CaLB. The Section first describes a set of use cases where clustering can be applied (Section \ref{SS:usecases}). Then, the system model used hereafter is stated in Section \ref{SS:model} and the general expressions of the required resources in uplink and downlink are developed in Section \ref{SS:required_resources}. Based on these expressions, the optimal clustering problem aimed to minimize the total number of resources is formalized in Section \ref{SS:optimal}. Finally, due to the complexity of the optimization problem, eCORE is proposed in Section \ref{SS:near_optimal} as a low complexity algorithm and CaLB is introduced in Section \ref{CaLB} to further enhance the capacity.

\subsection{Use cases}
\label{SS:usecases}

The clustering proposal addresses three use cases: the service of users located in coverage gaps, the enhancement of spectral efficiency and the load balancing (between cells and/or bands). Fig. \ref{fig:scenario_use_cases} sketches the initial scenario with 6 UEs served by one of the BSs (Fig. \ref{fig:baseline_scenario}) and the following cases:
\begin{itemize}
\item Extension of the coverage (Fig.\ref{fig:coverage_gap_scenario}): Assume that UE5 is in a coverage gap. If the quality of the links UE5-UE4 and UE4-BS2 is good enough, the clustering of UE4 (cluster head) and UE 5 can guarantee the service of the latter.
\item Spectral efficiency enhancement (Fig.\ref{fig:spectrum_efficient_scenario}): Clustering UE5 and UE3 with UE4 (the cluster head) increases spectral efficiency if: i) the quality of links UE3-UE4, UE5-UE4 and UE4-BS2 is significantly better than the quality of links UE5-BS2 and UE3-BS2; ii) downlink is highly loaded while uplink is less loaded.
\item Load balancing (Fig.\ref{fig:load_balanced_scenario}): If BS2 is highly loaded and BS1 is less loaded, the clustering of UE3 with UE2 (cluster head) can balance the load of BS2 to BS1. 
\end{itemize}

\begin{figure}[!ht]
    \centering
    \subfigure[Without clustering ]{\includegraphics[width=.4\linewidth]{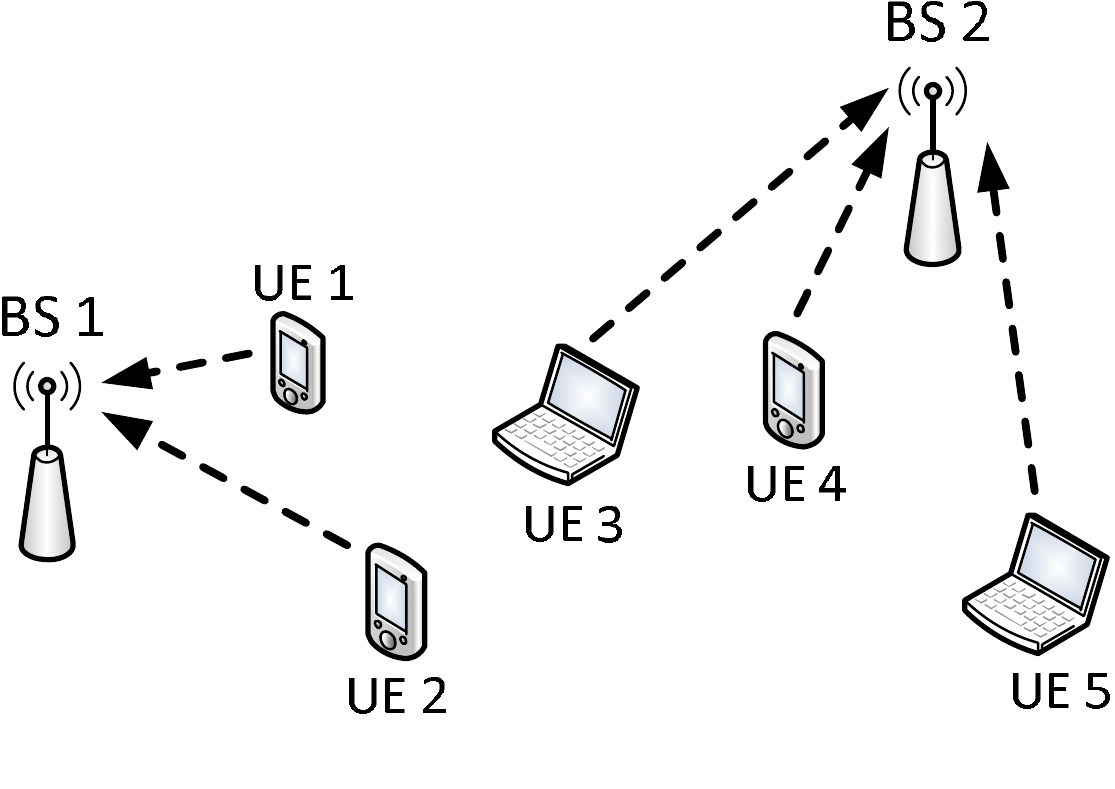} 	\label{fig:baseline_scenario}  } 
    \subfigure[Extension of the coverage ]{\includegraphics[width=.4\linewidth]{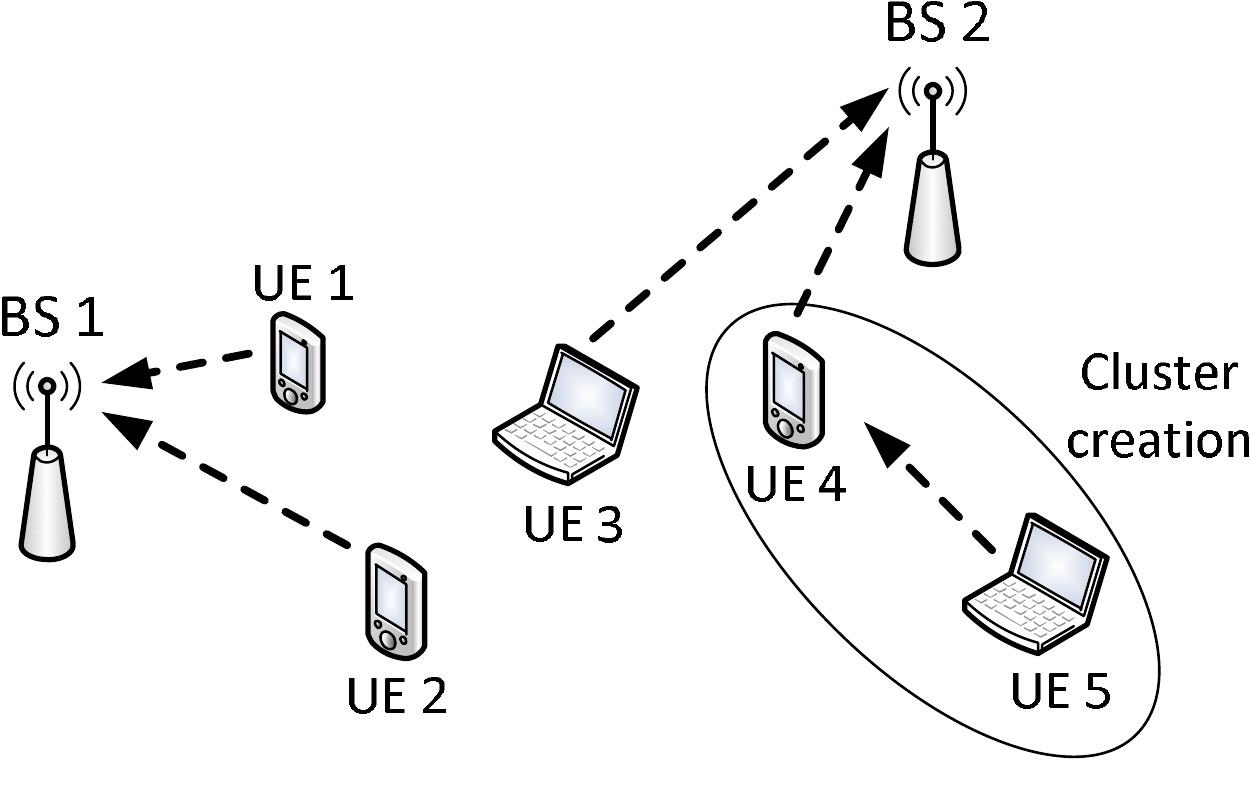} \label{fig:coverage_gap_scenario} } \\
    \subfigure[Spectral efficiency increase]{\includegraphics[width=.4\linewidth]{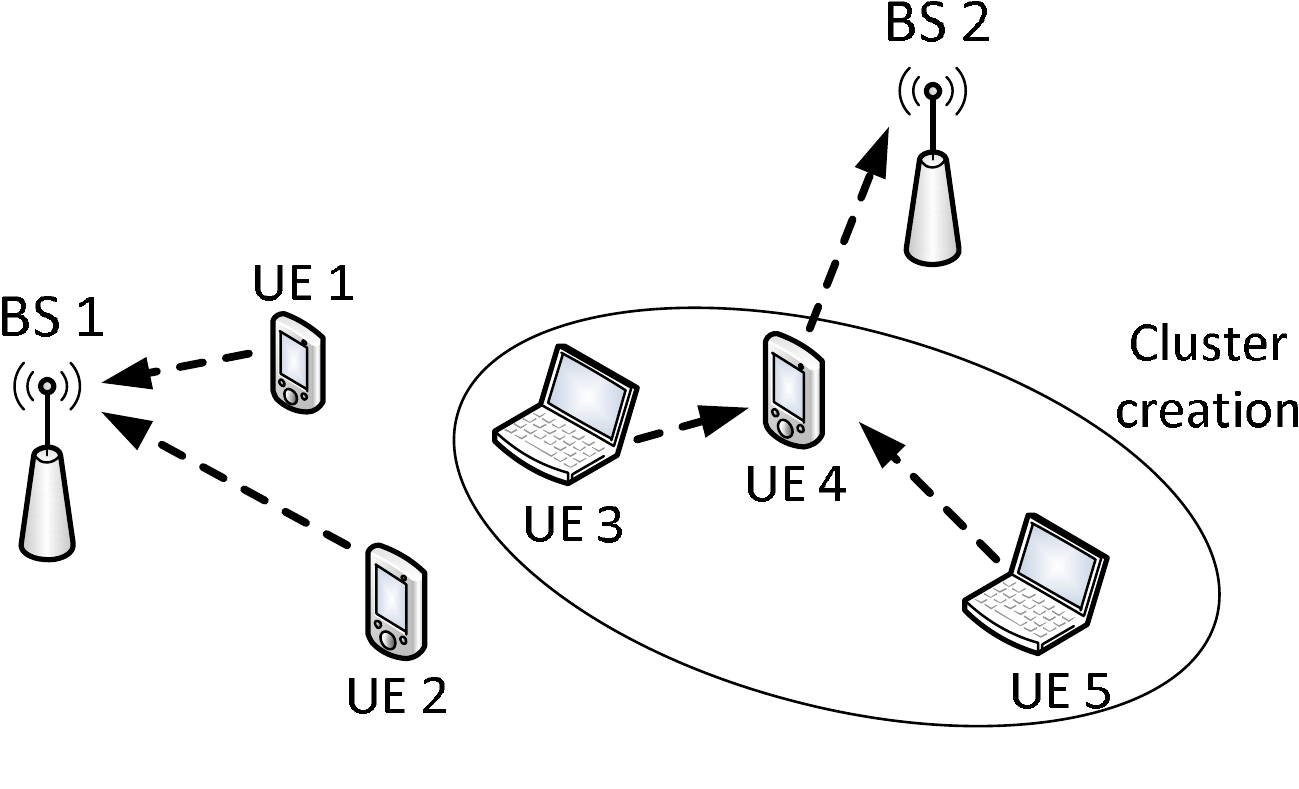} 	\label{fig:spectrum_efficient_scenario}  } 
    \subfigure[Load balancing]{\includegraphics[width=.4\linewidth]{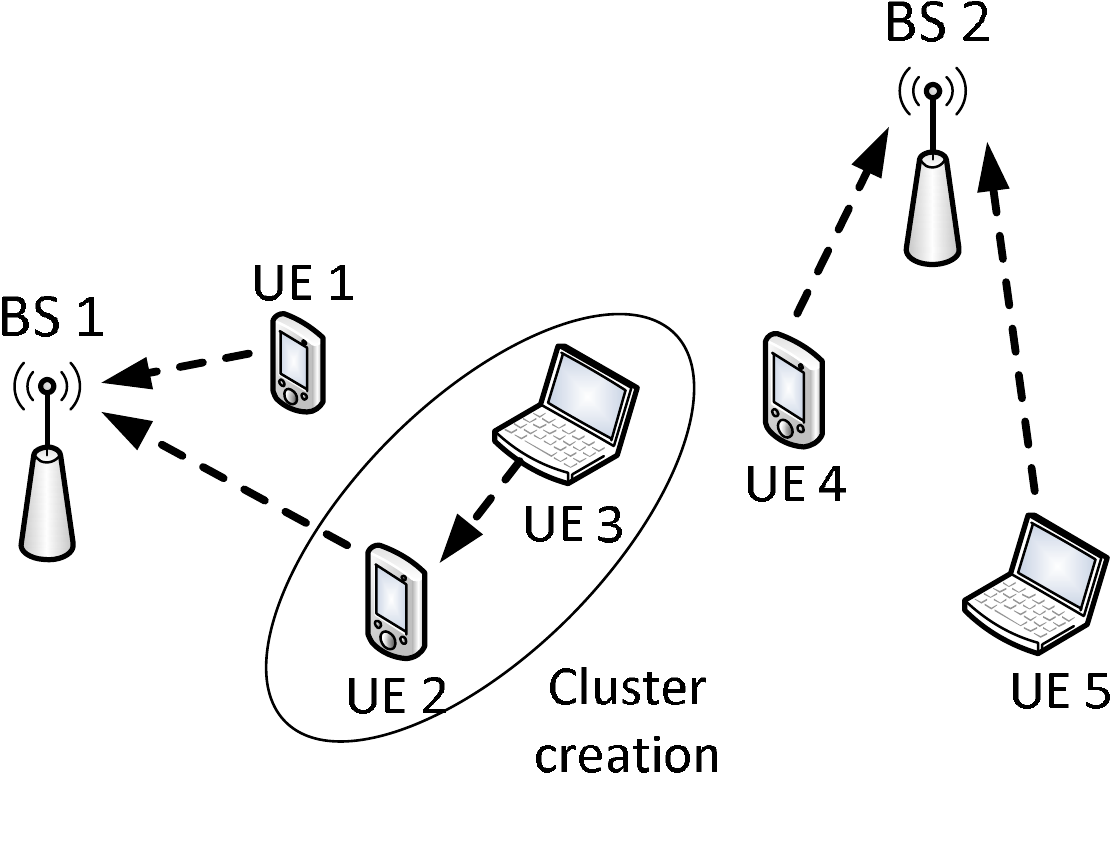} \label{fig:load_balanced_scenario} }
\caption{Example of possible clustering use cases}
\label{fig:scenario_use_cases}
\end{figure}

\subsection{System Model}
\label{SS:model}

The network is composed of a set of FDD-LTE BSs (macro eNBs and/or SCs), namely $\mathcal{B}$, covering the scenario and serving a set of users, denoted by $\mathcal{U}$. Each user $i \in \mathcal{U}$ is connected to a BS $k \in \mathcal{B}$ according to any of the existing cell association algorithms, such as the algorithms based on Reference Signal Received Power (RSRP) with or without Cell Range Expansion or on the Reference Signal Received Quality (RSRQ). The set of users connected to BS $k$ is referred to as $\mathcal{U}_k$. As users are assumed not to be served by more than one BS simultaneously, $\mathcal{U} = \bigcup_{k \in \mathcal{B}} \mathcal{U}_k$  and $\bigcap_{k \in \mathcal{B}} \mathcal{U}_k = \emptyset$. Each user $i \in \mathcal{U}$ is characterized by its traffic profile $\pi_i = (R_{i}^d, R_{i}^u )$, composed of the average transmission rate in the downlink $R_{i}^d$ and in the uplink $R_{i}^u$. 

As in general uplink and downlink traffic are unbalanced, $R_i^u = \alpha_i R_i^ d$, with $0 \leq \alpha_i \leq 1$. In LTE-A the transmission rate between two nodes depends on the selected Modulation and Coding Scheme (MCS), which is determined by the maximum allowed bit error rate (BER) and the SINR. Accordingly, the number of bits transmitted by user $i$ during a subframe time $T^{s}= 1$ms, defined as Transport Block Size (TBS), can be approximated by an attenuated and truncated form of Shannon bound. Thus, the TBS of  a transmission from $i$ to $j$ in the band $v$ ($v=u$ if the transmission is in the uplink band and $v=d$ if it is in the downlink band) may be approximatted as

\begin{equation}
\eta_{i,j}^v = T^{s}  r W \log_{2} (1+ \gamma_{i,j}^v)
\end{equation}  
\noindent where $r$ is the attenuating factor, $W$ is the bandwidth of a Physical Resource Block (PRB) and $\gamma_{i,j}^v$ is the SINR received at $j$ when data is transmitted by $i$. If the transmitter is a UE and the receiver is a BS,  $i \in \mathcal{U}$ and $j \in \mathcal{B}$; if the transmitter is a BS and the receiver is a UE,  $i \in \mathcal{B}$ and $j \in \mathcal{U}$; finally, if both transmitter and receiver are users in D2D mode, $i,j \in \mathcal{U}$.

\subsection{Resources required with and without clustering}
\label{SS:required_resources}

The spectral efficiency is measured in bps/Hz. Therefore, the enhancement of the spectral efficiency is equivalent to the minimization of the PRBs required to serve a given traffic. Based on the definitions stated above, the expected number of PRBs required in the scenario to serve all the users in the uplink ($N^u$) and in the downlink ($N^d$)  can be expressed as
\begin{equation}\label{Nd_no_cluster}
N^d = \sum_{k \in \mathcal{B}} N_k^d = \sum_{k \in \mathcal{B}} \sum_{i \in \mathcal{U}_k} \frac{R_i^d  T^s}{\eta_{k,i}^d} =  \sum_{k \in \mathcal{B}} \sum_{i \in \mathcal{U}_k} R_i^d  \phi_{k,i}^d
\end{equation}


\begin{equation}\label{Nu_no_cluster}
N^u = \sum_{k \in \mathcal{B}} N_k^u = \sum_{k \in \mathcal{B}} \sum_{i \in \mathcal{U}_k} \frac{R_i^u  T^s}{\eta_{i,k}^u} =  \sum_{k \in \mathcal{B}} \sum_{i \in \mathcal{U}_k} \alpha_i R_i^d  \phi_{i,k}^u
\end{equation}

\noindent where $N_k^d$ and $N_k^u$ are the expected number of PRBs per subframe required by base station $k$ (eNB or SC) in downlink and uplink. For simplicity, we define $\phi_{k,i}^d = \frac{ T^s}{\eta_{k,i}^d}$ and $\phi_{i,k}^u = \frac{T^s}{\eta_{i,k}^u} $.  


Let us consider that groups of users can create clusters. Each cluster $u$ is composed of cluster member users, among which a single user plays the role of cluster head. Hereafter, the set of users in cluster $u$ will be denoted by $\mathcal{C}_u$, and the cluster head by $h_u \in \mathcal{C}_u$. The cluster head $h_u$ is responsible for receiving the downlink traffic of all cluster members from the BS and forward it to the corresponding cluster member. Likewise, for the uplink traffic, the cluster head receives the traffic from the rest of the cluster members and forwards it to the BS. We will denote the set of all the clusters in the scenario by $\mathcal{C} = \bigcup_{u} \mathcal{C}_u$. Note that the communication within the cluster is carried out over the uplink band to minimize the  interference caused to the users outside the cluster. Therefore, intra-cluster communications are always carried out in the uplink band, whereas communications from/to the cluster head to/from the BS are both in the uplink band and in the downlink band. In real FDD networks, BSs are always full-duplex; conversely, user devices can be half-duplex (known as Half-Duplex FDD devices) or full-duplex (known as Full-Duplex FDD devices)\footnote{The term \textit{full-duplex} is defined as the ability of a node to transmit and receive simultaneously over uplink and downlink. The ability to transmit and receive simultaneously over the same band is not considered in this work.}. We also define the set of cluster heads as $\mathcal{H} = \{ h_u \}_{\forall u}$, and the set of cluster heads connected to BS $k$ as $\mathcal{H}_k = \mathcal{H} \cap \mathcal{U}_k$. Based on these definitions, the expected number of PRBs required in the downlink band ($\tilde{N}^d$) and in the uplink band ($\tilde{N}^u$) when clusters exist  are written as
\begin{equation} \label{Nd}
\tilde{N}^d = \sum_{k \in \mathcal{B}}\tilde{N}_k^d =\underbrace{\sum_{k \in \mathcal{B}} \sum_{i \in \mathcal{U}_k \setminus \mathcal{C}} R_i^d \phi_{k,i}^d}_{\text{non-clustered users}} + \underbrace{\sum_{k \in \mathcal{B}} \sum_{h_u \in \mathcal{H}_k} \phi_{k, h_u}^d \sum_{i \in \mathcal{C}_u} R_i^d}_{\text{clustered users}}
\end{equation} 

\begin{eqnarray} \label{Nu}
\tilde{N}^u &=& \sum_{k \in \mathcal{B}}\tilde{N}_k^u = \underbrace{ \sum_{\mathcal{C}_u \subseteq \mathcal{C}} \sum_{i \in \mathcal{C}_u \setminus \{ h_u \}} \left(  \phi_{i, h_u}^u R_i^u + \phi_{h_u ,i}^u R_i^d \right)}_{\text{transmissions within the cluster}} \nonumber \\
&+& \underbrace{\sum_{k \in \mathcal{B}} \sum_{i \in \mathcal{U}_k \setminus \mathcal{C}} R_i^u \phi_{i,k}^u}_{\text{non-clustered users}} + \underbrace{\sum_{k \in \mathcal{B}} \sum_{h_u \in \mathcal{H}_k} \phi_{ h_u,k}^u \sum_{i \in \mathcal{C}_u} R_i^u}_{\text{transmissions Cluster heads} \rightarrow \text{BSs}}  
\end{eqnarray}
\noindent where $\tilde{N}_k^d$ and $\tilde{N}_k^u$ denote the expected number of PRBs required by base station $k$ in the downlink and uplink, respectively, when clustering is applied. As it can be observed in \eqref{Nu}, intra-cluster communications do not interfere with uplink communications from the cluster head to the BS (they are not simultaneous). Likewise, it is worth noting that the number of PRBs required in the scenario is a function of the SINR, which in turn depends on the cell association algorithm. However, \eqref{Nd_no_cluster}-\eqref{Nu} are valid for a given SINR level and regardless of the cell association algorithm.

\subsection{Optimal clustering for spectral efficiency}
\label{SS:optimal}
The aim of the clustering technique presented herein is the minimization of the spectral resources utilization,  i.e. $\tilde{N} = \tilde{N}^u + \tilde{N}^d$.

As it can be observed, the minimization of the required resources is basically an association problem, where a user must be associated to a BS directly or through a cluster head. Let us define the association matrix $\mathbf{X} \in \{  0, 1 \}^{|\mathcal{U}| \times |\mathcal{B}|}$, where $| \cdot |$ is the cardinality  operator of a set, and the elements of $\mathbf{X}$ are $x_{i,k}=1$ if user $i$ is directly served by BS $k$ and  $x_{i,k} =0$ otherwise. Similarly, we define $\mathbf{Y} \in \{  0,1 \}^{|\mathcal{U}| \times |\mathcal{U}|}$ as the intra-cluster association matrix, with the elements of $\mathbf{Y}$ such that  $y_{j,i} = 1$ if user $j$ is connected to a BS through user $i$ (with $i$ playing the role of cluster head) and $y_{j,i} = 0$  otherwise. Using matrices $\mathbf{X}$ and $\mathbf{Y}$, the total number of required resources can be expressed as,
\begin{eqnarray} \label{N_matrix}
\tilde{N} (\mathbf{X},\mathbf{Y})  &=&  \sum_{i \in \mathcal{U}} \sum_{k \in \mathcal{B}} \Big[   x_{i,k} R_i^d \left(  \phi_{k,i}^d + \alpha_i \phi_{i,k}^u \right)   \\
&+&     \sum_{j \in \mathcal{U} \setminus \{ i \} }  y_{j,i} R_j^d \left(  \phi_{k,i}^d + \alpha_j \phi_{i,k}^u + \phi_{i,j}^u + \alpha_j \phi_{j,i}^u  \right) \Big]   \nonumber 
\end{eqnarray}

Therefore, the optimization problem is formulated as 
\begin{subequations} \label{opt_problem}
\begin{alignat}{2}
\min_{\mathbf{X},\mathbf{Y}} & \quad \tilde{N} (\mathbf{X},\mathbf{Y}) &   \tag{\ref{opt_problem}} \\
\text{s.t.} & \quad  x_{i,k},y_{i,j} \in \{ 0,1 \},& \quad \forall i,j \in \mathcal{U}, \forall k \in \mathcal{B} \label{constrain_a}\\
& \quad  \sum_{k \in \mathcal{B} } x_{i,k} + \sum_{j \in \mathcal{U}  } y_{i,j}=1,& \quad \forall i \in \mathcal{U} \label{constrain_c}\\
& \quad  \sum_{k \in \mathcal{B} }\sum_{i \in \mathcal{U}} x_{i,k} \geq 1, \label{constrain_d}\\
& \quad \sum_{k \in \mathcal{B}} x_{i,k}-y_{j,i} \geq 0, & \quad \forall i,j \in \mathcal{U} \label{constrain_e} \\
& \quad   y_{i,j} + y_{j,i} \leq 1,& \quad \forall i,j \in \mathcal{U} \label{constrain_f}\\
& \quad y_{i,i}=0, & \quad \forall i \in \mathcal{U} \label{constrain_g}
\end{alignat}
\end{subequations}

The optimization problem stated in \eqref{opt_problem} is an integer (binary) linear programming problem (ILP) \eqref{constrain_a}, where UEs can be served by either a BS or a cluster head \eqref{constrain_c} and at least one UE should be connected to a BS \eqref{constrain_d}. Moreover, a cluster can only be created if the cluster head is directly connected to a BS \eqref{constrain_e}, since multi-hops are not allowed within the cluster. 
A clustered user can either be a cluster head or be associated to a cluster head \eqref{constrain_f}. By definition, $y_{i,i}=0$ \eqref{constrain_g}.

\subsection{Enhanced Clustering Optimization for Resources Efficiency (eCORE)}
\label{SS:near_optimal}

As all 0-1 ILP problems are NP-hard \cite{ILP}, \eqref{opt_problem} is NP-hard. In order to overcome the complexity, a low complexity algorithm $(O(n^3))$, namely enhanced Clustering Optimization for Resources' Efficiency (eCORE), is presented. Based on the expressions derived in Section \ref{SS:required_resources}, some results can be enunciated.

\begin{lemma}\label{clustering_principle}
The number of resources required to serve a user $i \in \mathcal{U}_{k}$ is reduced when it joins a cluster with cluster head $j \in \mathcal{U}_{q}$ iff $( \phi_{k,i}^d - \phi_{q,j}^d - \phi_{j,i}^u  ) + \alpha_i ( \phi_{i,k}^u - \phi_{j,q}^u - \phi_{i,j}^u )>0$.
\end{lemma}
\begin{proof}
Lemma \ref{clustering_principle} is calculated from the difference between PRBs required in \eqref{Nd_no_cluster}-\eqref{Nu_no_cluster} and PRBs required in \eqref{Nd}-\eqref{Nu}. 
\end{proof}

\begin{lemma} \label{cluster_head_selection}
Given two users $i \in \mathcal{U}_k$ and $j \in \mathcal{U}_q$, the clustering gain $G_{i,j}$ when $j$ is the cluster head is defined as,
\begin{equation}\label{gain_definition}
G_{i,j} = R_i^d ( \phi_{k,i}^d - \phi_{q,j}^d - \phi_{j,i}^u  ) + \alpha_i R_i^d  ( \phi_{i,k}^u - \phi_{j,q}^u - \phi_{i,j}^u ) 
\end{equation}
The set of possible cluster heads of user $i$ is defined as $\mathcal{Y}_i = \{  j : G_{i,j}>0  \}$. For two users $i \in \mathcal{U}_k$ and $j \in \mathcal{U}_q$, if $\mathcal{Y}_i =\{ j \}$ and $\mathcal{Y}_j = \emptyset$, then $i$ and $j$ will create a cluster in which $j$ is the cluster head. Conversely, if $\mathcal{Y}_j \neq \emptyset$, $j \in \mathcal{Y}_i $ and $|\mathcal{Y}_i| >1 $, $i$ and $j$ will create a cluster where $j$ plays the role of cluster head if $G_{i,j} >G_{j,n} + G_{i,t}$ for $\forall n \in \mathcal{Y}_j$ and $\forall t \in \mathcal{Y}_i$. 
\end{lemma}

\begin{proof}
Using Lemma \ref{clustering_principle}, the clustering gain achieved by a cluster equals the aggregation of clustering gains of all cluster members. Thus, Lemma \ref{cluster_head_selection} can be derived from \eqref{Nd_no_cluster}-\eqref{Nu}. 
\end{proof}
According to Lemmas \ref{clustering_principle} and \ref{cluster_head_selection}, clustering is not limited to users within the same cell. Therefore, a cluster may be created by users previously connected to different cells (eNBs and/or SCs).

The proposed eCORE, described in Algorithm \ref{alg:the_alg}, is based on Lemmas \ref{clustering_principle} and \ref{cluster_head_selection} and it is aimed to create clusters that improve the total spectral efficiency. Therefore, the key parameter of the algorithm is the clustering gain ($G_{i,j}$) defined in Lemma \ref{clustering_principle}. eCORE starts with the computation of clustering gains for the different UEs, and initializing for each user $i$ the set $\mathcal{Y}_i$  of users $j$ that would result in a positive clustering gain, i.e. $G_{i,j}>0$ (line \ref{line1}). As stated in previous Sections, eCORE only considers single-hop intra-cluster communications to limit complexity and signalling. Accordingly, the term \textit{conflict} is used in the sequel to describe situations where a user $i$ has a positive clustering gain with a user $j$ ($G_{i,j}>0$) that, in turn, has a positive clustering gain with a third user $n$ ($G_{j,n}>0$). In these conflicting situations, either user $j$ becomes the cluster head of user $i$ or user $n$ becomes the cluster head of user $j$, but not both of them. Both situations are enunciated in Lemma \ref{cluster_head_selection} and implemented in Algorithm \ref{alg:the_alg}. Initially, eCORE clusters users without \textit{conflicts} to achieve the maximum clustering gain (line \ref{line3} to line \ref{line19}). In the second part, eCORE resolves the unsolved \textit{conflicts}, stored in the set $\mathcal{A}$ (see Algorithm \ref{alg:the_alg}), by selecting the option that provides the highest clustering gain (from line \ref{line21} to line \ref{line34}). 
\begin{algorithm}[h!]
\DontPrintSemicolon
\SetAlgoLined
\caption{Enhanced Clustering Optimization for Resources Efficiency (eCORE) \big(\textit{$O(n^3)$}\big)} \label{alg:the_alg}
\KwData{$\mathcal{U}$, $\phi_{k,i}^d$, $\phi_{i,k}^u$, $\phi_{i,j}^u$, $\phi_{j,i}^u$}
\KwResult{Set of Clusters $\mathcal{C} = \bigcap_u \mathcal{C}_u$}
Initialize the set of possible CHs ($\mathcal{Y}_i$), $\forall i \in \mathcal{U}$ \nllabel{line1}\\
$\mathcal{A}= \emptyset$: $\mathcal{A}$ is a set of UEs with $\mathcal{Y}_{i}\neq \emptyset$ \\
\For{$i \in \mathcal{U}$ \nllabel{line3}}{
\If{$\mathcal{Y}_i \neq \emptyset$}{$j^*=\underset{j}{\text{argmax }} (G_{i,j})$, $\forall j \in \mathcal{Y}_i$ and $G_{i}^{max} = G_{i,j^{*}}$\\
\eIf{$\mathcal{Y}_{j^*}=\emptyset$}{\eIf{UE $j^*$ is CH of cluster $u$}{$\mathcal{C}_u \leftarrow \mathcal{C}_u \cup \{ i \}$}{ $j^*$ is CH of a new cluster $u$: $\mathcal{C}_u =\{j^* , i\}$
}$\mathcal{Y}_{i}= \emptyset$\\
}{$\mathcal{A} \leftarrow \mathcal{A}\cup \{i\}$}
}
}\nllabel{line19}
UEs in $\mathcal{A}$ sorted in $G_{i}^{max}$ descending order \\
\While{$\mathcal{A}\neq \emptyset$ \nllabel{line21}}{
$i \leftarrow$ First UE in $\mathcal{A}$ \\
$j^*=\underset{j}{\text{argmax }} (G_{i,j}-G_{j,n}-G_{i,t})$, $\forall j,t \in \mathcal{Y}_i$ $\forall n \in \mathcal{Y}_j$ with \{   $\mathcal{Y}_t = \emptyset$ \textbf{or} $t\in \mathcal{A}$ \} \textbf{and} \{   $\mathcal{Y}_n = \emptyset$ \textbf{or} $n\in \mathcal{A}$ \} \\
\eIf{$G_{i,j^*} \leq 0$}{ 
$\mathcal{Y}_i = \emptyset$ and $\mathcal{A}  \leftarrow \mathcal{A} \setminus \{ i \}$ 
}{
\eIf{UE $j^*$ is CH of cluster $u$}{$\mathcal{C}_u \leftarrow \mathcal{C}_u \cup \{ i \}$ and $\mathcal{A}  \leftarrow \mathcal{A} \setminus \{ i \}$}{UE $j^*$ is CH of a new cluster $u$ ($h_u = j^*$)\\
$\mathcal{Y}_{j^*}= \emptyset$, $\>$ $\mathcal{C}_u =\{j^* , i\}$ and $\mathcal{A}  \leftarrow \mathcal{A} \setminus \{ i , j^*\}$}
}
}\nllabel{line34}

\end{algorithm}

eCORE manages to reduce the computational complexity by dividing the problem into two steps: the first step (lines 1-17) discards unfeasible clustering solutions, whereas the second step (lines 18-32) resolves conflicting cases. The first step is crucial to reduce the complexity, since it identifies potential cluster heads by figuring out if any of the associations would result in a reduction of the required resources. If not, that association is discarded (it is unfeasible for a spectral efficient cluster). In practice, the identification of potential cluster heads does not require a comparison of all users, since users farther than the D2D range can be discarded at the beginning.

In a nutshell, eCORE is an algorithm that checks which clusters can reduce the overall required PRBs. With this, not only the overall number of PRBs is reduced but traffic imbalance is decreased by transferring load from the downlink to the uplink.

\subsection{Clustering algorithm for Load Balancing (CaLB)}
\label{CaLB}

eCORE takes advantage of uplink and downlink traffic imbalance to decrease the downlink usage at the expense of an increase of the uplink usage (only if the downlink usage decrease is higher than the uplink usage increase). This fact limits, as it will be expounded hereafter, the maximum achievable capacity. Let us define the maximum number of PRBs allocated in the downlink and in the uplink to BS $k$ as $N^{d,max}_k$ and $N^{u,max}_k$. We can then define the saturation point of the cell (when the cell capacity reaches its limit) as the situation when either the downlink or the uplink cannot serve more traffic. Mathematically, the saturation is reached when

\begin{equation} \label{saturation}
\min (N^{u,max}_k - N^u_k , N^{d,max}_k - N^d_k  ) \approx 0
\end{equation}
\noindent where $N^u_k$ and $N^d_k$ are the PRBs used in each band in BS $k$ when there is no clustering. As traffic is generally more intense in downlink, when $N_k^d \gg N_k^u$ and $N^d_k \approx N^{d,max}_k$ it may be convenient to create clusters to increase the capacity even at the expense of a spectral efficiency decrease.

\begin{lemma} \label{load_balance_link}
Given a BS $k \in \mathcal{B}$ with an average number of required PRBs without clustering in the downlink and in the uplink equal to $N_k^d$ and $N_k^u$, respectively, the cell capacity is increased after creating the cluster $u$ (with $\mathcal{C}_u \subseteq \mathcal{U}_k$) if
\begin{equation}\label{link_cond}
\Delta N_k^u \leq \Delta N_k^d + (N_k^{u,max}-N_k^u)- (N_k^{d,max}-N_k^d)
\end{equation} 
\noindent even if the clustering gain is negative or null, i.e. $G_{\mathcal{C}_u}= \sum_{i \in \mathcal{C}_u \setminus \{ h_u \}} G_{i, h_u} = -(\Delta N_k^d+ \Delta N_k^u)  \leq 0$\footnote{According to \eqref{gain_definition}, the gain is defined as the reduction of the required PRBs, whereas $\Delta N_k^u$ and $\Delta N_k^d$ are defined as the increase of the required PRBs.}, where

\begin{equation} \label{delta_Nd}
\Delta N_k^d = \sum_{i \in \mathcal{C}_u \setminus \{ h_u \}} R_i^d \left( \phi_{k,h_u}^d - \phi_{k,i}^d \right)
\end{equation}
\begin{equation} \label{delta_Nu}
\Delta N_k^u = \sum_{i \in \mathcal{C}_u \setminus \{ h_u \}} R_i^d \left[    \phi_{h_u ,i}^u + \alpha_i \left( \phi_{h_u ,k}^u + \phi_{i, h_u}^u - \phi_{i,k}^u  \right) \right]
\end{equation}
\end{lemma}     
\begin{proof}
We define the number of available PRBs in the limiting band (the most loaded band) as  $A = \min (N^{u,max}_k - N^u_k , N^{d,max}_k - N^d_k  )$. If uplink is the limiting band, then $A=(N^{u,max}_k - N^u_k)$. Knowing that, by definition, $\Delta N_k^u >0$ and $\Delta N_k^d <0$, it can be found that \eqref{link_cond} is not true. Therefore, $A=(N^{d,max}_k - N^d_k)$ must be true (the downlink is more loaded). Rearranging \eqref{link_cond} we obtain that $(N_k^{d,max}-N_k^d - \Delta N_k^d \leq   N_k^{u,max}-N_k^u - \Delta N_k^u   )$, and therefore the number of available resources in the limiting band after clustering is $A'=N_k^{d,max}-N_k^d - \Delta N_k^d$. As $\Delta N_k^d <0$, then $A' > A$.
\end{proof}

\begin{lemma} \label{CaLB_conflict}
Given two users $i,j \in \mathcal{U}_k$, where user $i$ is not clustered and user $j$ is a cluster head, the number of PRBs required in the DL decreases when $i$ joins the cluster headed by $j$ if $\Delta N_k^d (i,j)<0$, with
\begin{equation}
\Delta N_k^d (i,j)= R_i^d (\phi_{k,j}^d - \phi_{k,i}^d)
\end{equation}
If $j$ is not clustered, and given two additional users $m$ and $n$ that minimize $x_{i,j} = \Delta N_k^d (i,j)-\Delta N_k^d (i,m)-\Delta N_k^d (j,n)$,
user $i$ must join the cluster headed by $j$ to maximize the reduction in the required PRBs if $x_{i,j}\leq 0$. Conversely, if $x_{i,j}> 0$, users $i$ and $m$ should create a cluster and users $j$ and $n$ should create a second cluster.
\end{lemma}
\begin{proof}
The first case is trivial, since $\Delta N_k^d (i,j)$ is, by definition, the increase in the downlink PRBs. If it is negative, the number of required PRBs decreases. If user $j$ is not a cluster head (second case), user $j$ can become the cluster head of user $i$ or the cluster member of an alternative cluster. In that case, if $n=\arg\min\limits_{q} \{ \Delta N_k^d(j,q) \}$ and $m=\arg\min\limits_{q} \{ \Delta N_k^d(i,q) \}$, the maximum overall reduction of PRBs would be $\Delta N_k^d(i,m)+\Delta N_k^d(j,n)$. Therefore, the maximum reduction of the PRBs in the downlink would result from clustering $i$ and $j$ if $\Delta N_k^d(i,j) <\Delta N_k^d(i,m)+\Delta N_k^d(j,n)$ (i.e. if $x_{i,j}<0$).  
\end{proof}

In order to further extend the capacity provided by eCORE, which is achieved by creating spectral efficient clusters, CaLB is proposed, mainly based on Lemmas \ref{load_balance_link} and \ref{CaLB_conflict}. It is aimed to improve the capacity when no additional spectral efficient clusters can be created, the downlink reaches the capacity limit and the uplink is still unloaded (see Algorithm \ref{alg:CaLB}). Therefore, CaLB is always run after the execution of eCORE. The inputs of CaLB are the set of users and clusters created by eCORE and two load thresholds, $n^d_{min}$ and $n^u_{min}$ for the downlink and uplink, respectively. These thresholds are used to determine whether a BS downlink and uplink are loaded or not as follows: if the number of available PRBs in the downlink, denoted in Algorithm \ref{alg:CaLB} by $n^d$ (see line \ref{CaLBline2}), is below $n^d_{min}$, the downlink of the BS is loaded; similarly, if the number of available PRBs in the uplink, denoted by $n^u$ (line \ref{CaLBline2}), is higher than $n^u_{min}$, the uplink of the BS is considered unloaded. Only in this case, each BS executes CaLB and triggers the clustering procedure (line \ref{CaLBline7}).

The algorithm establishes the clusters that reduce the load in the downlink, either by joining users to existing clusters or by establishing new clusters. To do that, all possible pairs of users (defined as $\mathcal{Q}_k$ in Algorithm \ref{alg:CaLB}) are ordered according to the reduction that would be achieved in the number of required downlink PRBs if clustered (i.e. $\Delta N_k^d (i,j)$). Note that there are constraints in this clustering process to prevent spectral efficient clusters (established by eCORE) from being destroyed. First, the cluster head of an existing cluster can serve new users by enlarging the cluster; that is, unclustered users can join existing clusters. A cluster head will not leave an existing cluster to become the cluster member of a new cluster. Finally, the clustering of a user must always result in a decrease of the downlink resources; therefore, 
the channel gain to the BS is higher for the cluster head than for the rest of cluster members ($\phi^d_{k, h_u } < \phi^d_{k, i }$ when a user $i$ joins a cluster head $h_u \in \mathcal{U}_k$). Based on the aforementioned constraints and on Lemma \ref{CaLB_conflict}, CaLB favours the clustering until the number of available PRBs in the downlink is larger than $n^d_{min}$ or the number of available PRBs in the uplink reaches the minimum, $n^u_{min}$.

\begin{algorithm}[h!]
\DontPrintSemicolon
\SetAlgoLined
\caption{Clustering algorithm for Load Balancing (CaLB) } \label{alg:CaLB}
\KwData{$n^d_{min},n^u_{min},  \{ \mathcal{U}_k, \mathcal{H}_k, N _k^{d,max}, N _k^{u,max}, \tilde{N}_k^d, \tilde{N}_k^u  \}_{ \forall k \in \mathcal{B}}$}
\KwResult{Set of Clusters $\mathcal{C} = \bigcap_u \mathcal{C}_u$ \\
}
\For{$k \in \mathcal{B}$ }{
$n^d = N_k^{d,max}-\tilde{N}_k^d$ and $n^u = N_k^{u,max}-\tilde{N}_k^u$ \nllabel{CaLBline2} \\
\If{$n^d <n^d_{min}$ }{
Define $\mathcal{Q}_k = \{ (i,j): \phi_{k,j}^d < \phi_{k,i}^d,  \forall i \in \mathcal{U}_k \setminus \mathcal{C}, \forall j \in (\mathcal{U}_k \setminus \mathcal{C}) \cup \mathcal{H}_k  \}$   \\
$(i,j) \in \mathcal{Q}_k$ are sorted in ascending order in $\mathcal{Q}_k$ based on $\Delta N_k^d (i,j)=R_i^d (\phi_{k,j}^d - \phi_{k,i}^d)$ 

\While{$\mathcal{Q}_k \neq \emptyset$ {\bf and}  $n^u \geq n^u_{min}$ {\bf and} $n^d < n^d_{min}$ \nllabel{CaLBline7}}{
$(i,j) \leftarrow$ First pair of nodes in $\mathcal{Q}_k$ \\ 
${\small \Delta N_k^u (i,j)=R_i^d \left( \phi_{j,i}^u + \alpha_i (\phi_{j,k}^u + \phi_{i,j}^u - \phi_{i,k}^u)      \right)} \normalsize$ \\
\eIf{$n^u+\Delta N_k^u (i,j) \geq \epsilon^u$}{

\eIf{$\exists u : j=h_u$}{
$\mathcal{C}_u \leftarrow \mathcal{C}_u \cup \{ i \}$ \\
$\mathcal{Q}_k  \leftarrow \mathcal{Q}_k \setminus \{ (i,m) : \forall m \neq i \}$ \\
$n^v \leftarrow n^v + \Delta N_k^v (i,j)$ for $v = \{ u,d\}$  \\
}{
Association according to Lemma \ref{CaLB_conflict} and update of $\mathcal{Q}_k$, $\mathcal{C}$, $n^d$ and $n^u$ 
}

}{
$\mathcal{Q}_k  \leftarrow \mathcal{Q}_k \setminus \{ (i,j) \}$
}
}

}
}

\end{algorithm}

To sum up, CaLB resumes the clustering process carried out by eCORE to create additional clusters. The created clusters are not spectral efficient, but reduce the uplink and downlink imbalance. CaLB is particularly appropriate when the downlink is highly loaded.

\section{Impact on Energy Consumption}
\label{S:energy}

Algorithms eCORE and CaLB rely on the set-up of cluster heads under the conditions stated in Section \ref{S:proposal}. However, the role of cluster head entails energy consuming tasks, e.g. receiving and retransmitting the data of the rest of cluster members. Therefore, the role of cluster head can cause early battery drain. In this section, the expression of the energy consumption of each stakeholder (i.e. cluster head, a cluster member and a non-clustered user) is derived, and the mitigation of possible energy overconsumption of the clustering approach is studied. In the following, the energy consumption expressions are derived in Section \ref{SS:energy_analysis}. In Section \ref{SS:optimal_energy} these expressions are used to modify the optimal clustering problem defined in Section \ref{SS:optimal} and to include energy overconsumption limits. Section \ref{SS:practical_approach} proposes a low complexity Clustering Energy Efficient algorithm (CEEa).

\subsection{Energy Consumption Analysis}
\label{SS:energy_analysis}

The energy consumption of a UE depends on two main factors: the Radio Resource Control (RRC) state of the device, that can be RRC\_CONNECTED or RRC\_IDLE, and the transmitted power \cite{VTC_power}. Let us define the RRC state space as $\mathcal{S} = \{  I, C_{tx}, C_{rx} \}$, where $I$ stands for the RRC\_IDLE state and the RRC\_CONNECTED state has been decoupled for convenience into two states, the transmitting state $C_{tx}$ and the receiving state $C_{rx}$. We also define $\mathcal{S}_C = \{ C_{rx}, C_{tx} \}$, i.e. $\mathcal{S}=\mathcal{S}_C \cup \{ I \}$. Based on this, the energy consumed by user $i$ during a subframe time $T^s$ is given by $E_i = T^s (P_{s_i} + P_{tx_{i}})$, where $P_{s_i}$ is the power consumed when user $i$ is in state $s_i \in \mathcal{S}$ and $P_{tx_{i}}$ is the transmitted power. The transmitted power differs in D2D mode (the intra-cluster communications) and in the communication with the BS. Hence, transmitted power of user $i$ is described in LTE \cite{36213} by,
\begin{equation}\label{Ptxi}
P_{tx_i} = \left\{
\begin{array}{ll}
M_i P_0 h_{i,k}^{-\xi} & \text{ if connected to BS } k \\
M_i P_{d2d} & \text{ if connected in D2D mode}  
\end{array}
\right.
\end{equation}

\noindent where $M_i$ is the number of PRBs scheduled for user $i$, $P_0$ is the target received power at BS $k$, $h_{i,k}$ is the channel gain between user $i$ and BS $k$, $\xi \in [0,1]$ is the compensating factor and $P_{d2d}$ is the transmitted power per PRB in D2D mode. For the sake of simplicity, in the following the role played by user $i$ is denoted by $\rho_i = \{ H,M,N \}$, with $\rho_i =H$ for a cluster head, $\rho_i =M$ for the rest of the cluster members and $\rho_i =N$ for the non-clustered users. Note that a user $i$ is directly connected to a BS if $\rho_i = \{H,N\}$, whereas it is in D2D mode if $\rho_i = M$. %

In the following, until explicitly mentioned, no mobility is considered. Therefore, each user is characterized by its profile $\pi_i$, the role $\rho_i$ and the location (channel gains with the rest of UEs and BSs), and the expected energy consumed during a subframe is expressed as
\begin{equation}\label{E_Ei}
\mathbb{E} [E_i  | \rho_i] = T^s \mathbb{E} [P_i  | \rho_i] = T^s  \mathbb{E}[P_{s_i} | \rho_i] + T^s \mathbb{E} [P_{tx_i}  | \rho_i] 
\end{equation}

\noindent where, by definition,

\begin{eqnarray}\label{expected_Psi}
\mathbb{E} [P_{s_i} | \rho_i]  &=& \mathbb{P} \{ s_i = I  | \rho_i\} P_I +  \mathbb{P} \{ s_i \in \mathcal{S}_C  | \rho_i\}  P_C \nonumber \\
&=& \mathbb{P} \{ s_i \in \mathcal{S}_C  | \rho_i\} ( P_C - P_I ) + P_I
\end{eqnarray}
\noindent where $P_I$ is the power consumed in state $s_i = I$ and $P_C$ is the power consumed in state $s_i \in \mathcal{S}_C$. Note that the probability of being in state $s_i$ depends on the role of the user. For instance, $\mathbb{P} \{ s_i = I | \rho_i = H  \} \leq \mathbb{P} \{ s_i = I | \rho_i = N  \}$. Taking into account that the cluster head forwards both the uplink traffic of all cluster members to the BS, and the downlink traffic to the cluster members (intra-cluster communications in D2D mode), the expected transmitted power of a user $i$ connected either to BS $k$ or to cluster head $h_u$ can be easily found using \eqref{Ptxi}.
\begin{equation} \label{E_Ptxi}
\mathbb{E} [P_{tx_i} | \rho_i] = \left\{
\begin{array}{ll}
P_0 h_{i,k}^{-\xi} \alpha_i R_i^d \phi_{i,k}^u & \text{ if } \rho_i = N \\
P_{d2d} \alpha_i R_i^d \phi_{i, h_u}^u  & \text{ if } \rho_i = M \\
P_0 h_{i,k}^{-\xi} \phi_{i,k}^u  \displaystyle \sum_{j \in \mathcal{C}_u }    \alpha_j R_j^d  + & \\
+ P_{d2d} \displaystyle \sum_{j \in \mathcal{C}_u \setminus  \{i \}} R_j^d  \phi_{i,j}^u  & \text{ if } \rho_i = H 
\end{array}
\right.
\end{equation}

\subsection{Optimal clustering with energy consumption constraints}
\label{SS:optimal_energy}

As mentioned, cluster heads tend to experience higher energy consumption than the rest of users. In order to limit the energy consumed by the cluster head, the problem defined in \eqref{opt_problem} must be modified to include the energy consumption constraint. If we define $w >0$ as the maximum allowed increase of the expected power/energy of a cluster head, the expected power consumed by a cluster head should not exceed the power consumed if it was not clustered:
\begin{equation}\label{power_condition}
\mathbb{E} [ P_i | \rho_i=H ] \leq (1+w) \mathbb{E} [ P_i | \rho_i=N ]
\end{equation}

As shown in \eqref{E_Ei}-\eqref{E_Ptxi}, the total power depends on the probability $\mathbb{P} \{  s_i \in \mathcal{S}_C  | \rho_i \}$ and on the transmitted power. Regarding the former, when the user $i$ is the cluster head, the probability can be divided into two components: the probability of $s_i \in \mathcal{S}_C$ due to the time required to transmit/receive its own traffic from/to the BS $k$ ($\theta_{i,k}^{N}$) and due to the time required to forward the traffic of the rest of the cluster members ($\theta_{i,j,k}^{H}$, for all users $j$ in the cluster).
\begin{equation} \label{P_si_divided}
\mathbb{P} \{  s_i \in \mathcal{S}_C | \rho_i = H\} = \sum_{k \in \mathcal{B}} (  x_{i,k} \theta_{i,k}^{N} + \sum_{j \in \mathcal{U}}  y_{j,i} \theta_{i,j,k}^{H} ) 
\end{equation}

\noindent where $x_{i,k}=1$ if user $i$ is served by BS $k$ and $x_{i,k}=0$ otherwise; and $y_{j,i}=1$  when user $i$ acts as the cluster head of user $j$ and $y_{j,i}=0$ otherwise (expressions for $\theta_{i,k}^{N}$ and $\theta_{i,j,k}^{H}$ are derived in Appendix \ref{S:definitions_theta}). By using \eqref{E_Ei}-\eqref{E_Ptxi} and \eqref{P_si_divided}, the components of \eqref{power_condition} can be written as
\begin{equation}
\mathbb{E} [ P_i | \rho_i= N ] = \theta_{i,k}^{N} \Delta P_{CI}  + P_I + \mathbb{E} [ P_{tx_i} | \rho_i = N  ]
\end{equation}
\begin{eqnarray}
{\small \mathbb{E} [ P_i  | \rho_i = H]}=\displaystyle \small P_I  + \sum_{k \in \mathcal{B}} x_{i,k} \left( \Delta P_{CI} \theta_{i,k}^{N} + \mathbb{E} [ P_{tx_i} | \rho_i = N  ] \right) \nonumber \\
+\sum_{j \in \mathcal{U}} y_{j,i} ( \theta_{i,j,k}^{H} \Delta P_{CI} + \mathbb{E} [ P_{tx_i} | \rho_i = H, j ])  
\end{eqnarray}
\noindent where $\Delta P_{CI}=P_C-P_I$ and $\mathbb{E} [ P_{tx_i} | \rho_i = H, j  ]$ is the power consumed by the cluster head attributable to the traffic of cluster member $j$, and it is defined as
\begin{equation}
\mathbb{E} [ P_{tx_i} | \rho_i = H, j  ] = R_j^d \left( P_0 h_{i,k}^{-\xi} \phi_{i,k}^u \alpha_j  + P_{d2d} \phi_{i,j}^u \right)
\end{equation}

Parameter $w$ must be selected to limit the energy overconsumption of cluster heads while allowing the creation of clusters. For instance, if only a 5\% power increase is allowed ($w=0.05$), cluster heads will not suffer from rapid battery drain but, in many cases, the establishment of some clusters will be compromised. Therefore, the optimization problem constrained by the energy consumption of the cluster heads results from including \eqref{power_condition} as a constraint into \eqref{opt_problem}.

\subsection{Clustering Energy Efficient algorithm (CEEa)}
\label{SS:practical_approach}

Due to the complexity of the optimization problem, in this Section we present a low complexity algorithm, namely CEEa, to manage the different energy consumption of each user. Note that, in principle, the energy consumed by a cluster head is higher than the energy consumed by a non-clustered user. Therefore, the energy consumption is clearly a disincentive for users to become cluster heads, even when $w$ is small. In a scenario without mobility, this disincentive can hardly be addressed (they can only be limited, as proposed in Section \ref{SS:optimal_energy}), but the changing environment offered by mobility opens up new possibilities. In order to analyse these possibilities, in the sequel the analysis is carried out as a function of time. 

Let us define the observation period $T_{\varepsilon}$ as the time during which the energy consumption is analysed to prevent users from energy overconsumption. For each user $i$, $T_{\varepsilon}$ can be divided into subperiods $T_{i,n} = [t_{i,n}^0, t_{i,n}^1) \in \mathbb{R}^2$ during which the role of user $i$ remains constant, i.e. $\rho_i (t_{i,n}^0) \neq \rho_i (t_{i,n}^0 - \delta t)$ for $\delta t \rightarrow 0$, and $t_{i,n}^1 = \max \{t : \rho_i (t) = \rho_i (t_{i,n}^0) , t > t_{i,n}^0  \}$. Based on the definitions, the time during which each user plays a specific role is the aggregation of periods with the same $\rho_i (t)$. Thus, three sets of periods $\mathcal{T}_i^H$, $\mathcal{T}_i^M$ and $\mathcal{T}_i^N$ are defined as $\mathcal{T}_i^m = \{ T_{i,n} : \rho_i ( t_{i,n}^0) = m  \}$ for $m = \{ H,M,N \}$.
If we denote the power consumed by user $i$ at time $t$ with role $\rho_i(t) = m$ as $P_i^m (t)$, and the power that would have been consumed by user $i$ at time $t$ in case of not being clustered as $\tilde{P}_i^N (t)$, the energy consumed over a subperiod $T_{i,n} \in \mathcal{T}^m$ with $m=\{ H,M \}$ and the energy that would have been consumed if $\rho_i (t) = N$ are given by $E_i^m(T_{i,n})= \int_{T_{i,n}} P_i^m (t) dt$ and $\tilde{E}_i^N(T_{i,n})= \int_{T_{i,n}} \tilde{P}_i^N (t) dt$ (the estimate of $\tilde{E}_i^N(t)$  can be found in Appendix \ref{S:EiN}). 
If the definition of energy overconsumption, namely $w(T_{\varepsilon})$, is given by $E_i^m(T_{\varepsilon})=(1+w(T_{\varepsilon}))\tilde{E}_i^N(T_{\varepsilon})$, it can be rewritten as
\begin{equation}
w(T_{\varepsilon}) = \displaystyle\frac{\sum_{T_{i,n} \in (\mathcal{T}_i^H \cup \mathcal{T}_i^M) } E_i^{\rho_i (t_{i,n}^0)} (T_{i,n})  }{\sum_{T_{i,n} \in (\mathcal{T}_i^H \cup \mathcal{T}_i^M) } \tilde{E}_i^{N} (T_{i,n}) }-1
\end{equation}

As, by definition, $P_i^H (t) > \tilde{P}_i^N (t) > P_i^M (t)$, user $i$ experiences energy overconsumption due to clustering if $w (T_{\varepsilon}) > 0$. Although the theoretical objective is to keep the total overconsumption around 0 in the long-term,  i.e. $\lim \limits_{T_{\varepsilon} \rightarrow \infty}w (T_{\varepsilon}) \approx 0$, in practice overconsumption must be limited over finite periods of time to avoid early battery drain. 
In the following CEEa is proposed to limit such overconsumption.

CEEa (see Algorithm \ref{alg:CEEa}) limits the overconsumption of users involved in the cluster by setting a maximum overconsumption threshold, referred to as $w_{max}$, that cannot be exceeded along the observation period $T_{\varepsilon}$. This observation period is divided into a set of $n_{\varepsilon}$ subperiods of duration $t_{\varepsilon}$, such that $T_{\varepsilon}=n_{\varepsilon}t_{\varepsilon}$. Specifically, for a given set of users, CEEa creates a list of users that cannot become cluster heads due to excessive energy consumption in the past, denoted by $\mathcal{Z}$, which is included as a constraint in eCORE. The maximum overconsumption condition, i.e. $E_i^m(t)>(1+w_{max})\tilde{E}_i^N(t)$, is checked at the end of each subperiod of duration $t_{\varepsilon}$ in two different ways: first, the energy consumption condition is checked for the total time since the beginning of the observation period (line \ref{CEEaline5}); secondly, the condition is checked for the actual subperiod (line \ref{CEEaline6}). Despite experiencing a total overconsumption, the user is not banned from remaining as cluster head if such overconsumption is not also experienced in the current subperiod, since it means that the overconsumption is starting to be compensated. Likewise, if the time during which the user has had the role $\rho_i=M$ until time $t$, $\tau_i^M (t)$, is smaller than the time during which it has had $\rho_i=H$ until time $t$, $\tau_i^H (t)$, the user cannot be cluster head. This condition works proactively to cope with situations where the cluster head suffers from slight but constant overconsumption. CEEa aims to compensate the overconsumption within $T_{\varepsilon}$. Therefore, the threshold $w_{max}$ is reduced at every observation subperiod with a factor ($\frac{n_{\varepsilon}-1}{n_{\varepsilon}}$) as the observation period draws on, since the higher $n_i$ is, the more difficult to compensate the energy consumption in the remaining $n_{\varepsilon}-n_i$ subperiods is.

Although there is not apparent incentive for a user to become cluster head in the short-term, this is not actually true. In loaded scenarios, not only cell-edge users can benefit from the proposed clustering, but also most of the users (even the cluster heads themselves, since the depletion of resources can impact on the resources allocated to them). In this context, CEEa eliminates the disincentive to become cluster head. The detection of selfish users is out of the scope of CEEa, but the proposed clustering algorithm does not preclude the design and implementation of additional algorithms running on top of CEEa to prevent selfish behaviours.

\begin{algorithm}[h!]
\DontPrintSemicolon
\SetAlgoLined
\caption{Clustering Energy Efficient Algorithm (CEEa)} \label{alg:CEEa}
\KwData{$\mathcal{U}$, $n_i\in [1 \ldots n_{\varepsilon}]$ for $\forall i \in \mathcal{U}$}
\KwResult{Set of users banned as cluster heads: $\mathcal{Z}$}
Initialization: if  $n_i= 1, \forall i \in \mathcal{U} \Rightarrow w_i = w_{max}$; $E_i=0$; $\tilde{E}_i=0$ \\
\For{$i \in \mathcal{U} : \rho_i(n_i t_{\varepsilon})= \{ H,M \}$}{
$ t_i=[(n_i-1) t_{\varepsilon},n_i t_{\varepsilon}]$ and $m= \rho_i(n_i t_{\varepsilon})$\\
\eIf{$E_i  > (w_i +1) \tilde{E}_i$ \nllabel{CEEaline5}}{
\uIf{$E_i^m (t_i) > (w_i +1) \tilde{E}_i^N (t_i) $ \nllabel{CEEaline6}}{
$\mathcal{Z} \leftarrow \mathcal{Z} \cup \{ i \}$
}\uElseIf{$\tau_i^{H} (n_i t_{\varepsilon}) > \tau_i^{M} (n_i t_{\varepsilon})$}{
$\mathcal{Z} \leftarrow \mathcal{Z} \cup \{ i \}$
}\Else{
$\mathcal{Z} \leftarrow \mathcal{Z} \setminus \{ i \}$
}
}{
$\mathcal{Z} \leftarrow \mathcal{Z} \setminus \{ i \}$
}
$w_i \leftarrow w_i \left(  \frac{n_{\varepsilon}-1}{n_{\varepsilon}} \right)$ and $n_i \leftarrow n_i +1$ \\
$E_i \leftarrow E_i +E_i^m (t_i)$ and $\tilde{E}_i \leftarrow \tilde{E}_i +\tilde{E}_i^N (t_i)$
}
\end{algorithm}

\section{Numerical Results}
\label{S:results}

\subsection{Scenario}
\label{SS:scenario}

In this section the proposed algorithms are validated and compared with existing algorithms found in the literature and with the results when no clustering algorithms are implemented (labelled in figures as \textit{Without Clustering} or \textit{w/o Clust.}). A custom-made simulator implemented in C++ has been used to simulate a network, which consists of a central eNB (macro BS) and the first interfering ring of 6 eNBs, with and inter-site distance of 500m. Under the coverage area of each eNB, 4 small cells are randomly deployed. The minimum distance between the eNB and a SC is 125m and the minimum inter-SC distance is 25m \cite{36842}. All eNBs are equally loaded and simulated, but only results from the central eNB and the corresponding 4 small cells are collected. Results are averaged over 1000 iterations. Users move at a constant speed of 3 km/h (pedestrian). The hit and bounce technique is used when users move out of the scenario under analysis: that is, when the user reaches the edge, it moves back into the scenario with a random direction \cite{36839}. 50\% of the deployed users are characterized by VoIP traffic (symmetric traffic with 64 kbps in the downlink and in the uplink) while the rest of users demand FTP or streaming traffic (700 kbps in the downlink).

The system is FDD and spectrum resource partition is considered between eNBs and SCs: eNBs and SCs operate in different bands \cite{36872}. No interference coordination techniques are considered in the simulations, and the PRBs are allocated randomly among users. Although interference coordination could lead to higher SINR levels (and better results), it has been omitted to better characterize the performance of the proposed algorithms. Moreover, interference coordination techniques are transparent for the proposed algorithms. Both users and the BS have a single antenna (SISO), and the spectral efficiency look-up table has been obtained from \cite{book2}. The rest of the parameters can be found in Table \ref{Simulation Parameters} \cite{36814}.

\begin{table}[!ht]
\centering
\caption{Simulation Parameters}
\label{Simulation Parameters}
\begin{tabular}{c|c}
\hline
 Parameter &  Value\\
\hline\hline
Bandwidth & Macro: 10 MHz \& Small Cell: 5 MHz \\
\hline
Macro cell Path-Loss   & $128.1 + 37.6  log_{10}$(distance km)   \\ 
\hline
Small cell Path-Loss   & $140.7+ 36.7  log_{10}$(distance km)    \\ 
\hline
D2D Path-Loss          & $148+40log_{10}$(distance km)           \\ 
\hline 
Max. BS Transmission power & Macro: 46 dBm \& Small Cell: 27 dBm \\
\hline
Max. UE Transmission power & Cellular: 20 dBm \& D2D: 18 dBm \\
\hline
\end{tabular}
\end{table}

\subsection{Results}
\label{SS:Results}

The objective of the optimal clustering for spectral efficiency stated in Section \ref{SS:optimal} (problem \eqref{opt_problem}) and labelled in figures as \textit{Optimal Clustering}, is the minimization of the total number of PRBs required to serve the traffic (i.e. the maximization of the spectral efficiency). Similarly, the optimal clustering with energy consumption constraints, detailed in Section \ref{SS:optimal_energy} and labelled hereafter as \textit{Energy Constrained}, is aimed to minimize the required PRBs while imposing energy overconsumption constraints for cluster heads. The spectral efficiency (bps/Hz) of these two solutions can be observed in Fig. \ref{fig:utilization} for 60 users, along with the results for our previous work CORE \cite{globecom}, and the proposed eCORE, CaLB (with  $n^d_{min} = 0.2N^{d,max}_k$, $n^u_{min} = 0.1N^{u,max}_k$) and CEEa (with $w=0.2$). It can be seen that the \textit{Optimal Clustering} is able to increase the spectral efficiency in the downlink band by clustering users and, therefore, by exploiting the good quality of the link between the BS and the cluster head. For instance, spectral efficiency in the downlink band rises a 54\% (from 1.26 bps/Hz to 1.95 bps/Hz) when \textit{Optimal Clustering} is applied in a scenario with 60 users. Although clustering solutions incur in additional PRBs utilization in the uplink band due to intra-cluster communications, it can also be observed that the total spectral efficiency (uplink and downink) increases. Therefore, the higher uplink band utilization is overcompensated by the downlink improvement. As it will be seen in Fig. \ref{fig:throughput}, when no clustering solution is applied, cell-edge users are not served due to low spectral efficiency. 
Fig. \ref{fig:utilization} also shows the spectral efficiency of \textit{Energy Constrained} when maximum energy overconsumption of the optimal clustering is limited to 10\% and to 50\% (i.e. $w=0.1$ and $w=0.5$, respectively). As expected, the overconsumption constraint prevents clusters from being set up if they result in excessive energy overconsumption. Thus, only clusters that are simultaneously spectral efficient and that keep cluster heads consumption below a threshold (i.e. $w$) are set up. This is the reason why the spectral efficiency is lower as the energy constraint becomes more restrictive (i.e. a lower $w$). For instance, the downlink spectral efficiency is 1.27 bps/Hz when $w=0.1$ and 1.36 bps/Hz when $w=0.5$. 
Some insights can be found in Table \ref{t:Clusters}, where the average number of clusters and the average size of each cluster are shown for 30 and 60 users. In the \textit{Energy Constrained} solution, the reduction of $w$ (lower overconsumption is allowed) has a higher impact on the number of clusters created than in the size of the cluster. That is, whereas the size of the cluster remains stable, overconsumption constraints cause a significant reduction in the average number of clusters. 

Besides the results of the optimization problems stated in Sections \ref{SS:optimal} and \ref{SS:optimal_energy}, Fig. \ref{fig:utilization} also includes the results for CORE, eCORE, CaLB and CEEa. eCORE achieves results (in terms of spectral efficiency) very close to the optima, with a performance less than 5\% lower than  \textit{Opt. Clust.} 
Moreover, eCORE manages to increase the downlink spectral efficiency with respect to CORE, since it enables the establishment of clusters among users from different cells. 

Table \ref{t:Clusters} shows that the intensification in the creation of clusters promoted by eCORE results in the setup of more clusters, although with a similar size. For instance, for 60 users eCORE doubles the number of clusters with respect to CORE while the average size of each cluster is approximately the same.

Something similar occurs with CaLB: the number of clusters grows more than the average size of the clusters. That is, CaLB creates new clusters rather than enlarge the clusters initially established by eCORE. However, note that CaLB enables the creation of non-spectral efficient clusters if the imbalance between uplink and downlink is thereby reduced. This is the reason why although the downlink spectral efficiency in CaLB is higher than in eCORE, the opposite occurs with the total spectral efficiency (uplink and downlink bands). Finally, as CEEa limits the energy consumption by deterring some users from being/remaining cluster heads, the spectral efficiency is reduced with respect to eCORE and CaLB. Table \ref{t:Clusters} also shows that the number of clusters is reduced due to the energy consumption constraints.
    
\begin{figure}[!ht]
\centering
\includegraphics[scale=.6]{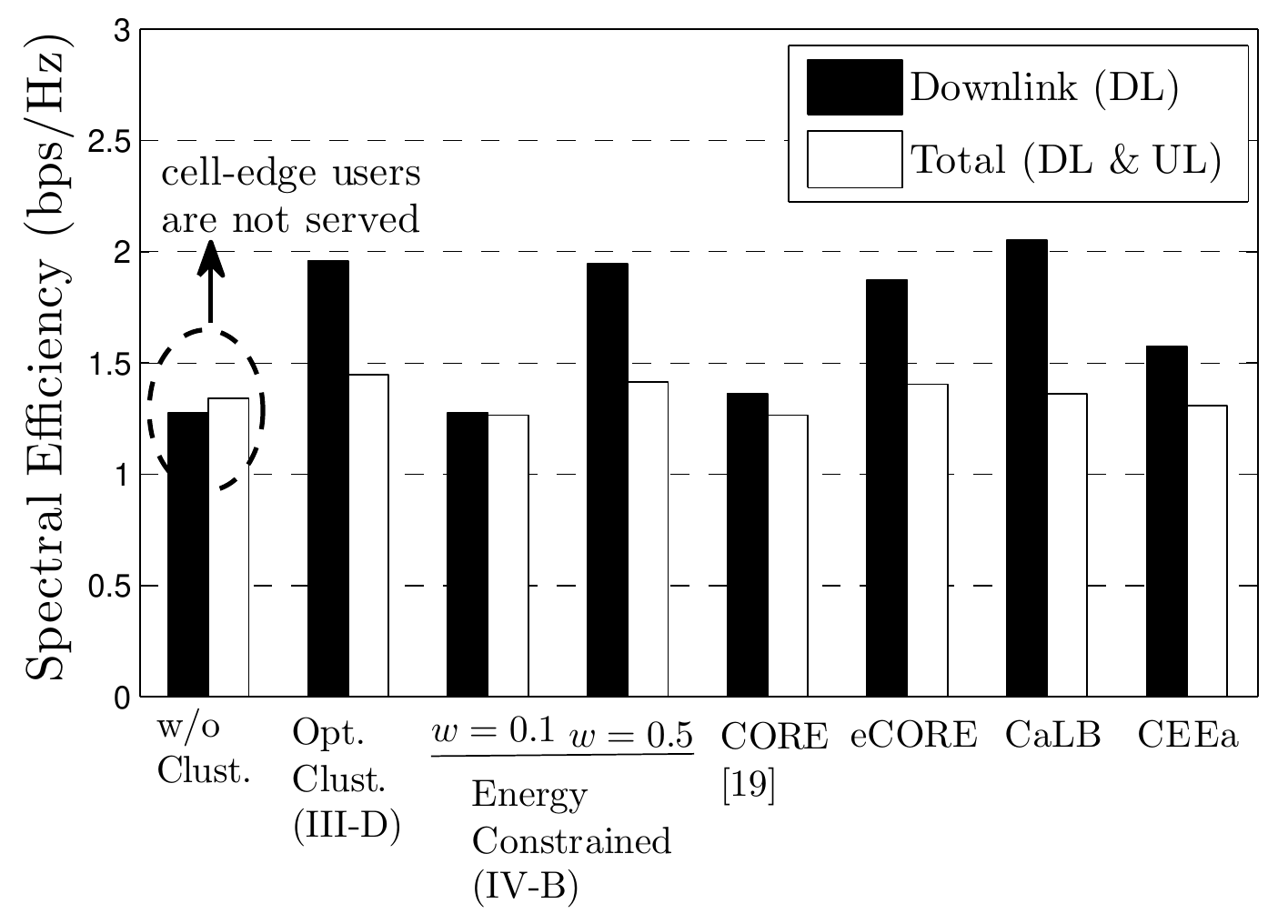} 
\caption{Downlink and total spectral efficiency for 60 users.}
\label{fig:utilization}
\end{figure}

\begin{table}[!ht]
\centering
\caption{Average number and size of clusters}
\label{t:Clusters}
\begin{tabular}{c|cc|cc}
\hline
 & \multicolumn{2}{c|}{Avg. Num. Clusters} & \multicolumn{2}{c}{Avg. Cluster Size} \\ \hline 
Num users & 30 & 60 & 30 & 60 \\ \hline \hline
Optimal Clustering & 5.38	& 11.45 & 2.37 & 2.67 \\
Energy Constrained ($w$=0.1) & 1.44	& 3.27 & 2.29 & 2.88 \\
Energy Constrained ($w$=0.5) & 2.68	& 5.21 & 2.35 & 2.85 \\
CORE & 5.45	& 11.39	& 2.42 & 2.69 \\
eCORE & 5.71 & 11.64 & 2.43 & 2.75 \\
CaLB & 5.71	& 15.27	& 2.43 & 2.75 \\
CEEa & 3.38	& 7.62 & 2.50 & 2.98 \\
\hline
\end{tabular}
\end{table}

Although the impact of the proposed algorithms on the spectral efficiency has been analyzed, Fig. \ref{fig:throughput} shows the downlink throughput for each algorithm. Fig. \ref{fig:throughput} also includes as baseline the algorithm  proposed in \cite{bench}, which is labelled as CS.  As CS is a scheme based on the received SNR to allow or ban cooperation (among other aspects), results for two minimum SNR thresholds have been simulated: 4.73 dB and 2.84 dB. Fig. \ref{fig:throughput} shows how CaLB outperforms the rest of algorithms, reaching a 59.5\% gain in the downlink throughput with respect to the case \textit{Without Clustering} for 140 users. As expected, it can be also observed that eCORE outperforms CORE and, in turn, CaLB outperforms eCORE. In particular, CORE achieves a throughput 36.6\% higher than \textit{Without Clustering}, whereas eCORE reaches a 47.2\% improvement and CaLB a 59.5\%. As for CEEa, the additional constraints imposed in the creation of clusters reduce the downlink achievable throughput, but still presents slightly better results than CORE.
\begin{figure}[!ht]
\centering
\includegraphics[scale=.55]{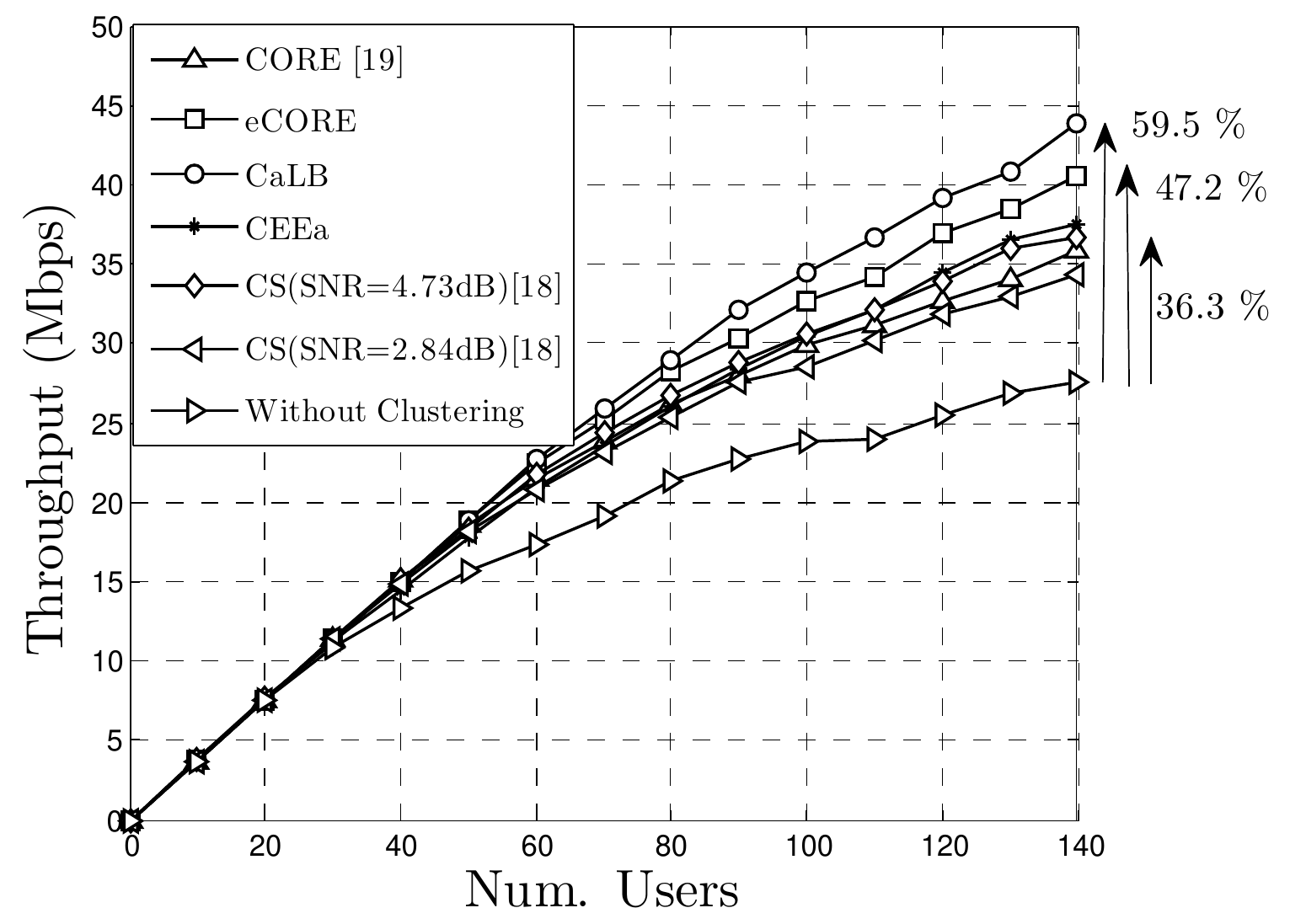} 
\caption{Downlink throughput for the set of algorithms.}
\label{fig:throughput}
\end{figure}

Focusing on how CEEa is able to limit the energy overconsumption suffered by cluster heads, Fig. \ref{fig:CDF_total} plots the Cumulative Distribution Function (CDF) of the energy overconsumption, $w$, for eCORE, CaLB and CEEa. Note that, by definition, the overconsumption is always expressed with respect to the case where no clustering algorithms are implemented. Therefore, without any clustering, the energy overconsumption would be $w=0\%$. As it can be observed in Fig. \ref{fig:CDF_total}, the energy underconsumption from which cluster members (except for the cluster head) benefit is similar in eCORE, CaLB and CEEe. However, CEEa limits the overconsumption of cluster heads. For instance 99\% of the users have an overconsumption $w< 20\%$ with CEEa; in turn, for eCORE the 99\% of users experience an overconsumption $w< 240\%$ and with CaLB the same percentage of users experience $w<260\%$. Therefore, CEEa is able to limit the overconsumption of cluster heads.

\begin{figure}[!ht]
    \centering
    \includegraphics[scale=.55]{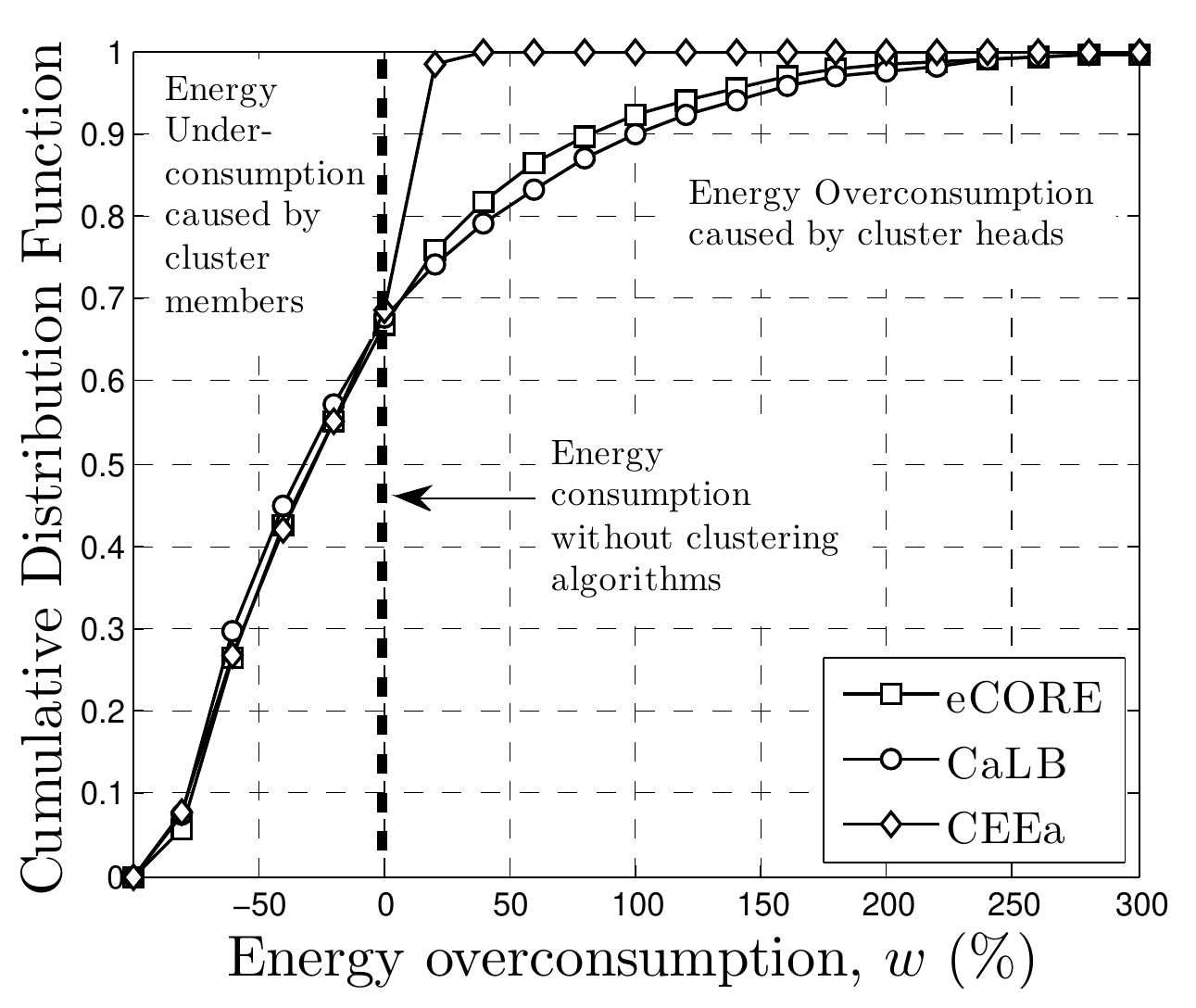} 	 
\caption{Cumulative Distribution Function (CDF) of the energy overconsumption with 60 users.}
\label{fig:CDF_total}
\end{figure}

Given the trade-off between the maximum capacity gain (achieved by CaLB) and the minimum impact on energy overconsumption (achieved by CEEa), Fig. \ref{fig:EnEff} sheds light on the energy efficiency of eCORE, CaLB and CEEa for 60 users. Cluster heads present low energy efficiency because they forward traffic to/from cluster members. Therefore, the percentage of users with low energy efficiency grows with the number of cluster heads. This can be particularly significant in CaLB and eCORE. Conversely, CEEa alleviates partially the high energy consumption of cluster heads but, simultaneously, makes the throughput decrease. In none of the cases (eCORE, CaLB and CEEa) the low energy efficiency of cluster heads is  compensated by the increased energy efficiency of the rest of cluster members. Accordingly, and in the light of the results, clustering algorithms can improve the capacity of the network but at the expense of lower energy efficiency.

\begin{figure}[!ht]
\centering
\includegraphics[scale=.55]{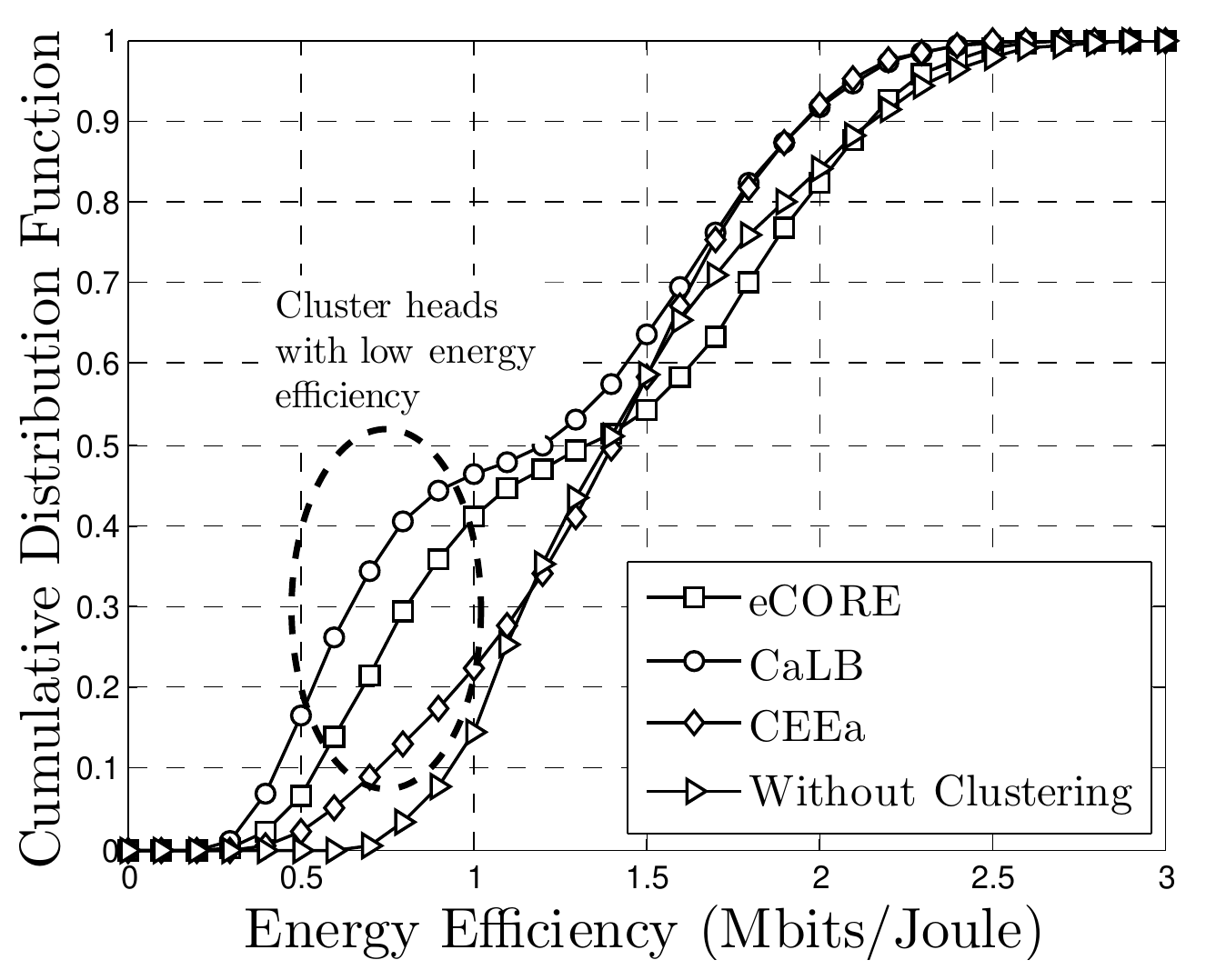} 
\caption{Cumulative Distribution Function (CDF) of the energy efficiency with 60 users.}
\label{fig:EnEff}
\end{figure}

In order to see how sensitive CaLB and CEE are to their key parameters ($n_{min}^d$ and $n_{min}^u$ for CaLB and $w$ for CEEa), simulations have been run with different values. As for CaLB, it has been observed that differences in terms of throughput are not significant and below 2\% for a wide range of values $n_{min}^d$ and $n_{min}^u$. Although the creation/enlargement of clusters will start before as the values of $n^d_{min}$ increase, it is also true that it will not be translated into a significant increase of the throughput. Therefore, CaLB is slightly sensitive to $n^d_{min}$ variations in terms of throughput as long as $n^d_{min}>0$, but should be selected small enough to avoid the creation of additional clusters when it is not actually needed (in terms of throughput)\footnote{No additional figure for the throughput  has been included due to the slight observed  differences}.	
		
Regarding CEEa, the key parameter is the maximum allowed energy overconsumption $w$. This parameter has a single objective that is attained in a two-fold manner: firstly, by preventing some users from becoming cluster heads (due to previous energy overconsumption), and secondly by forcing the release of the role of cluster head (if the energy overconsumption is too high). In a nutshell, the larger $w$ is, the more aggressive the clustering is, thus achieving similar results to the ones obtained with eCORE (where no energy consumption constraints are imposed). Conversely, small $w$ values impose additional constraints in the creation of clusters. This effect can be observed in Fig. \ref{fig:ReEff2}, where the CDF of the energy efficiency is plotted for 60 users and $w=\{ 0.2 , 0.6, 1.5 \}$.  Results for eCORE have been also included for the sake of comparison. It can be clearly observed that eCORE has cluster heads with low energy efficiency and in turn cluster members with high energy efficiency. Note that the higher $w$ is, the more closed results are to the ones of eCORE, since less constraints on energy consumption are imposed.

\begin{figure}[!ht]
\centering
\includegraphics[scale=0.55]{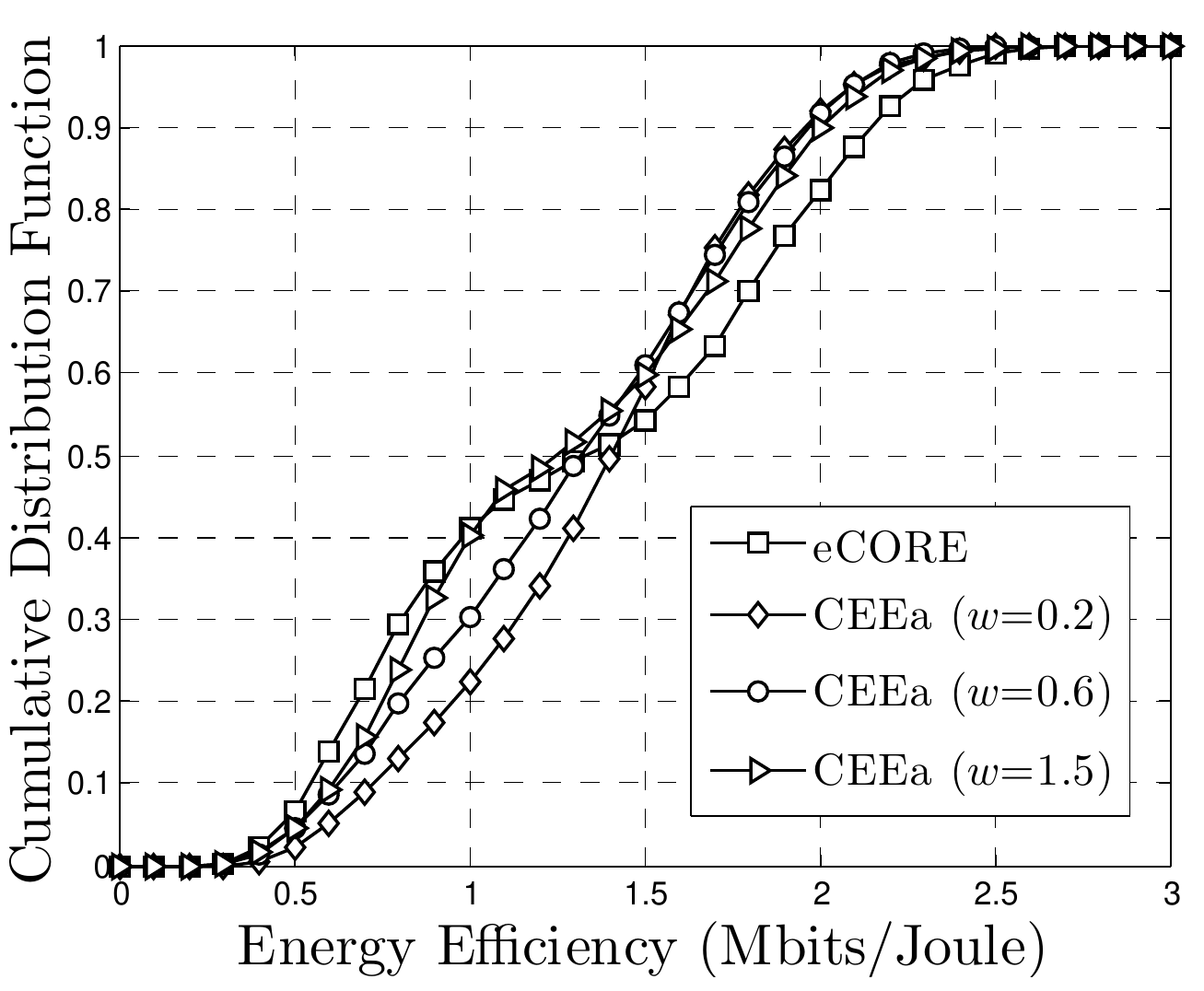} 
\caption{Cumulative Distribution Function (CDF) of the energy efficiency of CEEa with 60 users for different $w$ values.}
\label{fig:ReEff2}
\end{figure}

\subsection{Discussion on signalling}
\label{SS:signalling}

Signalling is an important aspect of D2D communications. 3GPP establishes control and data plane paths for D2D communications (termed as Proximity Services -ProSe) in \cite{22803}, and covers these aspects in more detail in \cite{23303}. The proposed algorithms are framed within the group of \textit{UE-to-Network Relay} functions \cite{23303}, since the cluster head acts as a relay from each of the cluster members to the network. In this context two important interfaces are defined: PC3, defined as the interface from the relay (i.e. the cluster head) to the network; and PC5, defined as the one-to-one or one-to-many interface between users (the so-called D2D communication). The proposed mechanisms implement the network-assisted D2D mode with the loosely-controlled scheme, in which the network allocates resources for the D2D communications, and the cluster head reallocates the resources within the cluster. Network-assisted loosely-controlled D2D communications require additional signalling, particularly over PC5 interface. However, as shown in Table \ref{t:Clusters}, the proposed algorithms improve the throughput by creating a significant number of small size clusters rather than large size clusters, thus alleviating/reducing the increase of signalling over the PC5 interface. Therefore, although eCORE, CaLB and CEEa require additional signalling, the small size of the clusters limits the additional signalling burden over PC5.

Nevertheless, frequent cluster head (re-)selection could incur excessive signalling burden. Thus, there exists a trade-off between signalling and system performance. Algorithms eCORE and CaLB do not include neither parameters to control (reduce or increase) the number of clusters nor parameters to limit the duration of the clusters. Conversely, CEEa controls indirectly the number and size of the clusters, as well as how long they remain active or with the same cluster head, with parameters $w_{max}$ and $T_{\varepsilon}$.

\section{Conclusions}
\label{S:conclusions}

This work presents a complement/alternative to the costly densification of cellular RANs based on the creation of clusters of users, where intra-cluster communications are carried out in a D2D mode. Three clustering algorithms are presented: eCORE,  CaLB and CEEa. eCORE is aimed to optimize the usage of spectral resources by establishing spectral efficient clusters. Due to the significant imbalance between uplink and downlink traffic, CaLB is an algorithm that creates non-spectral efficient clusters that, however, improve the maximum capacity of the network by reducing the aforementioned imbalance. Finally, and in order to reduce the impact of eCORE on the energy consumption of cluster heads, CEEa is proposed to keep track of the overconsumption of users and ban some users from becoming cluster heads. Results show that the proposed clustering solutions increase the capacity of the network. In particular, the most aggressive clustering algorithm (CaLB) outperforms the rest of algorithms. Yet, it has be shown that any capacity improvement is translated into an increase of the consumed energy or, in other words, a reduction of the energy efficiency. In that sense, CEEa achieves a good energy consumption performance but, simultaneously, it leads to the smallest capacity gain. 

\appendices

\section{Calculation of $\theta_{i,k}^{N}$ and $\theta_{i,j,k}^{H}$ }
\label{S:definitions_theta}

The following expressions are derived for full-duplex devices, i.e. devices that can transmit in the uplink and receive in the downlink simultaneously \cite{half_full1}. Expressions for half-duplex devices are omitted due to space limitation, but they can be easily derived. Regarding full-duplex devices, the probability of being in RRC\_CONNECTED state when the scheduler minimizes the RRC\_CONNECTED state periods can be expressed as\footnote{In these expressions all probabilities are conditioned to $\rho_i =N$ or $\rho_i =H$ respectively, but it has been omitted to simplify the notation.} $\mathbb{P}\{ s_i \in \mathcal{S}_C | \rho_i = N\} = \max \left( \mathbb{P} \{ s_i = C_{tx} \}, \mathbb{P} \{ s_i = C_{rx}  \} \right)$ and $\mathbb{P}\{ s_i \in \mathcal{S}_C | \rho_i = H\}=  \max ( \mathbb{P}  \{ s_i = C_{rx} \cap v=d \},\mathbb{P} \{ s_i = C_{tx} \} +\mathbb{P}  \{ s_i = C_{rx} \cap v = u \})$,

where $v$ stands for the band ($u$ for uplink and $d$ for downlink) over which user $i$ is receiving. We define the maximum number of PRBs  per subframe allocated in the band $v$ ($v=u$ for uplink and $v=d$ for downlink) to a user served by a BS as $M_{max}^v$. If the user is a cluster head, then the maximum number of allocated PRBs is given by $M_{max}^{C,v}$ (the cluster head must forward all the traffic generated/received within the cluster). Analogously, the maximum number of PRBs per subframe that can be allocated to a user for intra-cluster communications are denoted by $M_{max}^{H}$ (transmissions from cluster head to cluster member) and $M_{max}^{M}$ (transmissions from cluster member to cluster head). Accordingly, it can be found that
\begin{equation}
\mathbb{P} \{ s_i = C_{rx} | \rho_i = N\} \approx \frac{R_i^d \phi_{k,i}^d}{M_{max}^d} \nonumber
\end{equation}
\begin{equation}
\mathbb{P} \{ s_i = C_{tx} | \rho_i = N \} \approx \frac{R_i^u \phi_{i,k}^u}{M_{max}^u} \nonumber
\end{equation}
\begin{equation}
\mathbb{P} \{ s_i = C_{rx} \cap v=d | \rho_i= H\} \approx \frac{ \phi_{k,i}^d  }{M_{max}^{C,d}}  \sum_{j \in \mathcal{C}_u } R_j^d \nonumber
\end{equation}
\begin{equation}
\mathbb{P} \{ s_i = C_{rx} \cap v=u | \rho_i= H \} \approx    \sum_{j \in \mathcal{C}_u \setminus \{ i \}  }  \frac{R_j^u \phi_{j,i}^u}{M_{max}^{M}}  \nonumber
\end{equation}

\begin{eqnarray}
\mathbb{P} \{ s_i = C_{tx} | \rho_i= H \} &\approx& \frac{R_i^u \phi_{i,k}^u}{M_{max}^{C,u}} +  \nonumber \\
&+& \sum_{j \in \mathcal{C}_u \setminus \{  i  \} } \left(  \frac{R_j^u \phi_{i,k}^u}{M_{max}^{C,u}} + \frac{R_j^d \phi_{i,j}^u}{M_{max}^{H}} \right) \nonumber
\end{eqnarray}

Based on this, $\theta_{i,k}^N$ can be expressed as $\theta_{i,k}^{N} =\max \left\{  \frac{R_i^d \phi_{k,i}^d}{M_{max}^d} , \frac{R_i^u \phi_{i,k}^u}{M_{max}^u}  \right\}$. The expression of $\theta_{i,j,k}^{H}$ can be derived analogously.

\section{Estimate of $\tilde{E}_i^N (t)$}
\label{S:EiN}

Following the notation used in Appendix \ref{S:definitions_theta} and according to \eqref{E_Ei} and \eqref{E_Ptxi}, the estimate of $\tilde{E}_i^N (T_{i,n})$ involves two components: the time during which user $i$ would remain in RRC\_IDLE state or in RRC\_CONNECTED state if it was not clustered, and the transmitted power. If full-duplex devices are assumed, the RRC\_CONNECTED time can be expressed as $T_i^{C} \approx T^s \cdot \max \left\{ \left\lceil   \frac{\eta^d_{tot} (T_{i,n})}{M_{max}^d \eta^d_{i,k}} \right\rceil ,  \left\lceil   \frac{\eta^u_{tot} (T_{i,n})}{M_{max}^u \eta^u_{i,k}} \right\rceil  \right\}$, 
where $\eta^v_{tot} (T_{i,n})$ is the total number of bits of user $i$ transmitted during period $T_{i,n}$ ($v=u$ for uplink traffic and $v=d$ for downlink traffic), $\eta^u_{i,k}$ would be the TBS if user $i$ was served by BS $k$, and $M_{max}^v$ is the maximum number of PRBs allocated to a user in a single subframe. Based on \eqref{E_Ei} and the expression for $T_i^{C}$, $\tilde{E}_i^N (T_{i,n}) \approx T_i^C P_C + T_i^I P_I + T^s P_0 h_{i,k}^{-\xi} \left\lceil   \frac{\eta^u_{tot} (T_{i,n})}{M_{max}^u \eta^u_{i,k}} \right\rceil$, 
where $P_{I}$ and $P_{C}$ are the power consumed in state $s_i = I$ and $s_i \in \mathcal{S}_C$ respectively, and $T_i^I =  \sum_{T_{i,n}} (t_{i,n}^1-t_{i,n}^0) - T_i^C$.

\ifCLASSOPTIONcaptionsoff
  \newpage
\fi



\end{document}